\documentclass[a4paper,twocolumn,11pt,accepted=2023-12-23]{quantumarticle}
\pdfoutput=1

\usepackage[utf8]{inputenc}
\usepackage[english]{babel}
\usepackage[T1]{fontenc}
\usepackage{amssymb,amsthm,mathtools,dsfont}
\usepackage{hyperref}
\usepackage{graphicx} 
\usepackage{physics}
\usepackage[capitalise,noabbrev]{cleveref}		
\usepackage{tikz}
\usepackage{enumitem}
\usepackage{orcidlink}

\newtheorem{lemma}{Lemma}
\newtheorem{corollary}[lemma]{Corollary}

\newenvironment{equation-aligned}
{
\begin{equation}
\begin{aligned}
}
{
\end{aligned}
\end{equation}
}

\DeclareMathOperator{\diag}{diag}
\DeclareMathOperator{\CZ}{CZ}
\DeclareMathOperator{\Cz}{\mathcal{CZ}}
\DeclareMathOperator{\CP}{C\Phi}
\DeclareMathOperator{\Cp}{\mathcal{CZ}_p}
\DeclareMathOperator{\CS}{\mathcal{CS}}

\newcommand{\hhd}{{\hat h^\dagger}}
\newcommand{\hvd}{{\hat v^\dagger}}

\DeclareMathOperator{\1}{\mathds{1}}
\DeclareMathOperator{\NN}{\mathcal{N}}

\makeatletter
\newlength{\negph@wd}
\DeclareRobustCommand{\negphantom}[1]{
  \ifmmode
    \mathpalette\negph@math{#1}
  \else
    \negph@do{#1}
  \fi
}
\newcommand{\negph@math}[2]{\negph@do{$\m@th#1#2$}}
\newcommand{\negph@do}[1]{
  \settowidth{\negph@wd}{#1}
  \hspace*{-\negph@wd}
}
\newcommand{\removed}[1]{}
\makeatother

\title{Useful entanglement can be extracted from noisy graph states}
\author{\orcidlink{0000-0001-7676-1605} Konrad Szymański} 
\affiliation{Naturwissenschaftlich-Technische Fakultät, Universität Siegen, Walter-Flex-Straße 3, 57068 Siegen, Germany}
\affiliation{Research Center for Quantum Information, Slovenská Akadémia Vied,
Dúbravská cesta 9, 84511 Bratislava, Slovakia}
\author{\orcidlink{0000-0003-2753-6027} Lina Vandr\'e}  
\affiliation{Naturwissenschaftlich-Technische Fakultät, Universität Siegen, Walter-Flex-Straße 3, 57068 Siegen, Germany}
 \affiliation{Technische Universität Wien, Atominstitut, Vienna Center for Quantum Science and Technology, Stadionallee 2, 1020 Vienna, Austria}
\author{\orcidlink{0000-0002-6033-0867} Otfried Gühne}
\affiliation{Naturwissenschaftlich-Technische Fakultät, Universität Siegen, Walter-Flex-Straße 3, 57068 Siegen, Germany}
\date{28 October, 2025}

\begin{document}

\maketitle
\begin{abstract}
     Cluster states and graph states in general offer a useful model of the stabilizer formalism and a path toward the development of measurement-based quantum computation. Their defining structure -- the stabilizer group -- encodes all possible correlations that can be observed during measurement.  The measurement outcomes which are
     consistent with the stabilizer structure make error correction possible. 
     Here, we leverage both properties to design feasible families of states that can be used as robust building blocks of quantum computation. This procedure reduces the effect of experimentally relevant noise models on the extraction of smaller entangled states from the larger noisy graph state. 
     In particular, we study the extraction of Bell pairs from linearly extended graph states -- this has the immediate consequence for state teleportation across the graph. We show that robust entanglement can be extracted by proper design of the linear graph with only a minimal overhead of the physical qubits. This scenario is relevant to systems in which the entanglement can be created between neighboring sites.
     The results shown in this work 
     provide a mathematical framework for noise reduction in measurement-based quantum computation. With proper connectivity structures, the effect of noise can be minimized for a large class of realistic noise processes.
\end{abstract}

\section{Introduction}

The discovery of cluster states -- states of qubits with grid-like entanglement structure -- provided a new perspective for quantum computation, better applicable for some systems \cite{Raussendorf2001,Nielsen_2006_cluster_QC}. The original design consisted of qubits arranged in a two-dimensional square lattice structure, where the neighboring qubits would be entangled by means of controlled-$Z$ gates. 
Unlike the circuit-based understanding of quantum computation \cite{Mu_oz_Coreas_2022_Qcircuits,Fisher_2023}, the input state is constant, and the only available class of operations is the sequential measurement of individual qubits. It is \emph{the act of measurement} that feeds the information into the system and performs all of the needed transformations \cite{Raussendorf2003measurement}. Generalizations into different connectivity structures -- the graph and hypergraph states -- soon followed, sharing the same core properties with the original cluster states \cite{Rossi2013,hein2006entanglement,Kruszynska_2009,Qu_2013} and providing a framework for further research into the entanglement properties of multiqubit states \cite{Gachechiladze_2019,Vandre2023Entanglement}.
Another task for which graph states are useful is the distribution of entanglement in quantum networks   \cite{hahn2019quantum, Mannalath2023, Meignant_2019, sen2024multipartiteentanglementdistributionquantum, fan2024optimizeddistributionentanglementgraph,  basak2024improvedroutingmultipartyentanglement}.

In the context of quantum computation, the most basic operation is no processing at all: a simple transfer or teleportation  \cite{Bennett1993Teleporting,Bouwmeester_1997_exp_teleportation} of one qubit state into another place. It is a basic building block of measurement-based quantum computation, on top of which more involved operations are built. In the basic scenario, teleportation can be thought of as the creation of a Bell pair, which is subsequently used to transmit the quantum state. The Bell pair creation step can be performed by sequential measurement of Pauli-$X$ operators on a path between the initial and end qubits 
and Pauli-$Z$ on the qubits neighboring the path (see \cref{fig:clustertoBell}). Intuitively, the $Z$-basis measurement excises the entanglement contained in the path from the rest of the cluster state 
 and the $X$-basis measurement along the path leaves the terminal (first and last) qubits in Bell-type correlations, up to local unitary rotations. The final measurement of one of them fuses the two, and the surviving qubit shares correlations with a distant part of the graph (\cite{Raussendorf2003measurement}, see also \cref{app:bellteleportation}).

The crucial point is that the Bell pair creation step is independent of the state to be transferred. The right correlational structure is prepared and subsequently used, and the preparation step can be studied alone.  This abstraction step is helpful in various ways: the theoretical effects of the exact input state for teleportation are not relevant, this formulation allows us to utilize the works dedicated to quantum state extraction \cite{morruiz2023influence,Frantzeskakis23weightedgraphs}, and similar language regarding the correlation of the relevant qubits is frequently used \cite{Wagner_2023,Miller_2023}, also in the research dedicated to teleportation \cite{Morley-Short_2019_crazygraphs}.

 \begin{figure}%[t]
    \centering
    \includegraphics[width=\linewidth]{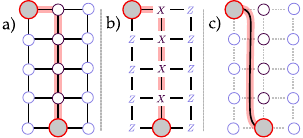}
    
    \caption{
    Any two vertices of a graph state associated with a connected graph can be turned into a Bell pair by sequential measurements. 
    a) A grid graph describing the underlying structure of a cluster state. The goal is to generate a Bell pair between the qubits associated with the big red vertices. The chosen path is highlighted in red.
    b) We measure the qubits associated with the vertices on the path between the red vertices in the $X$ basis and the qubits adjacent to the chosen path in the $Z$ basis. 
    c) The post-measurement state has Bell correlations between the qubits associated with the red vertices. This procedure 
    can be formalized to the so-called X-protocol \cite{hahn2019quantum}.
	}
    \label{fig:clustertoBell}
\end{figure}

In general, entanglement is required for quantum computation, and preparation of a Bell pair or GHZ state is an operational primitive of many quantum algorithms. Extracting those states from larger graph states was discussed in \cite{dejong2023extracting,MorRuiz_2023_noisy_weighted,Morley-Short_2019_crazygraphs,Mannalath2023,hahn2019quantum,freund2025graphstateextractiontwodimensional}. The problem of extracting graph states from larger graph states is known to be NP-complete \cite{dahlberg2018transformgraphstatesusing, Dahlberg2020transforminggraph}.

  In this article, we analyze experimentally motivated noise models and determine what can be done to the initial noisy graph state structure and the measurement pattern in order to mitigate the noise effects. This general problem of noise reduction in graph states has been analyzed from various viewpoints. For instance, purification and distillation protocols are discussed in \cite{Aschauer_2005, Kruszynska_2006, sajjad2024lowerboundsbipartiteentanglement}. Quantum error correcting codes are more general, but complex. Unlike them, our approach utilizes only local single-qubit measurements, basic preparation and verification, and is, in principle, adaptable to various underlying graph structures. However, it still uses the properties of the \emph{stabilizer} associated with the graph state. In the ideal case, some outcomes are perfectly (anti)correlated, and any deviation indicates the presence of noise. As mentioned above, teleportation is equivalent to a Bell pair preparation; in order to simplify the calculations, we analyze the entanglement quality of the prepared state after the internal qubits have been measured.

The article is structured as follows. In \cref{sec:prelim}, we briefly introduce graph states and the stabilizer formalism. In \cref{sec:extraction}, we explain the measurement scheme that we use to generate Bell pairs between certain nodes or for teleportation. We introduce noise models in \cref{sec:noise} and analyze how to deal with their effects in \cref{sec:noise_extract}. In \cref{sec:GHZ}, we show how to extend the presented methods to generate GHZ states. 
We conclude our work in \cref{sec:sum}. A portion of this article contains numerical results; the code \cite{zenodo} used in data generation and analysis is publicly available.

\section{Graph state preliminaries} \label{sec:prelim}

\begin{figure}[ht!]
    \centering
    \includegraphics[width=\linewidth]{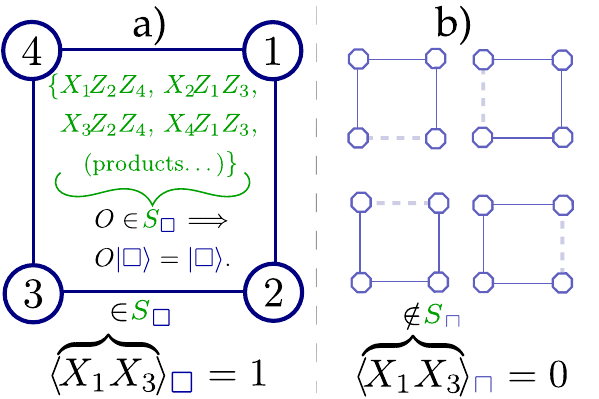}
    \caption{a) The square graph defines the associated graph state $\ket{\square}$ which, as any graph state, possesses an associated structure of the \emph{stabilizer} $S_{\square}$ composed of specific strings of Pauli operators. Expectation value of any Pauli string $O$ is nonzero if and only if one of $\pm O$ belongs to $S_{\square}$. $X_1 X_3$ is a stabilizer operator of $\ket{\square}$ and therefore it's expectation value is 1.
    \\b) Graph states corresponding to the state defined in a) being altered by noise -- here, the noise may remove certain edges (dashed lines), as in the exemplary $\ket{\sqcap}$ defined by the graph shown in the top left corner. Noise processes change the stabilizer group (e.g., by graph modification shown here), which allows for verifiable preparation.
    }
    \label{fig:fourqubit}
\end{figure}

In this section, we present a short introduction to the class of graph states and the stabilizer formalism.

\subsection{Simple graph states}

A graph $G = (V,E)$ is defined by a set $V$ of vertices and a set of edges $E \subset V^2$: each $\{i,j\}=:e \in E$ denotes a connection between two vertices $i, j \in V$, where $i \neq j$. 
Graph states are multi-qubit quantum states, where the vertices and edges of the graph $G = (V,E)$ represent qubits and entangling gates, respectively. The state $\ket{G}$, corresponding to a graph $G$, is defined as
\begin{align}
    \ket{G} =\negphantom{ii}\prod_{\{i,j\} \in E}\negphantom{ii} \CZ_{i,j} \ket{+}^{\otimes \lvert V \rvert }, 
    \label{eq:grstate}
\end{align}
where $\CZ$ denotes a standard controlled-$Z$ unitary operation on qubits $i$ and $j$
\begin{align}
    \CZ_{i,j} = \!\mathbbm{1}_{V\setminus\{i,j\}}\!\!\otimes\!\!\Big(&\ket{00}\!\!\bra{00}\!+\!\ket{10}\!\!\bra{10} \notag \\
    &\!\!+\!\ket{01}\!\!\bra{01}\!-\!\ket{11}\!\!\bra{11}\Big)_{i,j}.
\label{eq:czaction}
\end{align}
The controlled-$Z$ gates are diagonal in the computation basis, and hence commute. Therefore, the product in the above equation does not depend on the order of operations.

Consider a four-vertex square graph shown in \cref{fig:fourqubit}a: it has a vertex set $V^
\square = \lbrace 1,2,3,4 \rbrace$ and an edge set $E^\square = \lbrace \lbrace 1,2 \rbrace, \lbrace 2,3 \rbrace, \lbrace 3,4 \rbrace, \lbrace 4,1 \rbrace \rbrace$. The corresponding graph state is given by 
\begin{align}
    \ket{\square} = \CZ_{1,2} \CZ_{2,3} \CZ_{3,4} \CZ_{4,1} \ket{+}^{\otimes 4}.
\end{align}

The graph states can also be described using the so-called stabilizer operators, providing a description of the state in question in terms of its correlation structure. Consider an initial state $\ket{+}^{\otimes\vert V\vert}$: since it is a pure product state, there are no correlations between different sites. However, it is a $(+1)$-eigenstate for each $X_i$ for $i\in V$. Let us denote the product of controlled-$Z$ operators in  \cref{eq:grstate} by $U$; then
\begin{align}
    \ket{G}&=U \ket{+}^{\otimes \lvert V \rvert } =U X_i \ket{+}^{\otimes \lvert V \rvert } \\
   &= \underbrace{U X_i U^\dagger}_{g_i} \underbrace{U \ket{+}^{\otimes \lvert V \rvert }}_{\ket{G}}=g_i \ket{G}.
   \label{eq:stabgenorigin}   
\end{align}
The operator $g_i$  can be shown to have the form of
\begin{align}
g_i = X_i \negphantom{ii}\prod_{j \in \NN(i)}\negphantom{ii}Z_j,
\label{eq:generator}
\end{align}
where $X_i$, $Y_i$, and $Z_i$ denote the Pauli matrices acting on the $i$-th qubit and $\NN(i)$ is the neighborhood of the vertex $i$, given by
\begin{align}
    \NN(i) = \{ j : \{i,j\} \in E \}.
\end{align}
As shown above, any graph state $\ket{G}$ defined by \cref{eq:grstate} is an eigenstate with eigenvalue $+1$ of all $g_i$:
\begin{align}
    g_i \ket{G} = \ket{G}.
    \label{eq:stabgendef}
\end{align}
While the $(+1)$-eigenspace of each of  the operators $g_i$ is highly degenerated, taken together they fully define the state $\ket{G}$. The graph state $\ket{G}$ is the unique $(+1)$-eigenstate of all $g_i$.  Operators defined by \cref{eq:stabgendef} commute; hence, they generate an Abelian subgroup of the group of Pauli operators, called the \emph{stabilizer} $S$ of $\ket{G}$:
\begin{equation}
S = \left\{\prod_{i\in I} g_i: I \subset V\right\}.
\end{equation}
The elements of the group $S$, due to the compact form of the generators (\cref{eq:generator}), can be described with the help of low-dimensional matrices of integers modulo $2$ \cite{Wu_2016_Xchains}. This representation is helpful in direct calculations involving the stabilizer group.

The exemplary graph state shown in \cref{fig:fourqubit}a has stabilizer operators $g_1 = X_1 Z_2 \1_3 Z_4$, $g_2 = Z_1 X_2 Z_3 \1_4$, $g_3 = \1_1 Z_2 X_3 Z_4$, and $g_4 = Z_1 \1_2 Z_3 X_4$, as well as all products composed from these operators.

The elements of the stabilizer $S$ form a complete description of all possible correlations of local measurements of graph states.  Let $O=o_{i_1}o_{i_2}\ldots$ be a Pauli string, where $o_i$ denotes one of the Pauli operators $\{X_i, Y_i, Z_i\}$ acting on the $i$-th qubit. If $\pm O$ appears in $S$, its expectation value on $\ket{G}$ is equal to $\pm 1$. Otherwise, if $\pm O\notin S$, we have $\langle O\rangle_{G}=0$.

For example, for the graph state shown in \cref{fig:fourqubit}a, the string $O = X_1  X_3 = g_1 g_3$ is a stabilizer operator of $\ket{\square}$. Using \cref{eq:stabgendef} it can be seen that $\langle O\rangle_{\square}=1$. The string $O' = X_1 Z_3 \notin S_\square$, hence $\langle O' \rangle_{\square}=0$. In this paper, we are mostly using Pauli strings which contain only Pauli $X$ operators \cite{Wu_2016_Xchains}. However, the results can be generalized to other Pauli strings.

\subsection{Weighted graph states}

\emph{Weighted graph states} \cite{Hartmann_2007} are a related class of quantum states, for which controlled-phase operators are used instead of controlled-$Z$. 
The phases $\varphi_{i,j}$ are possibly different for each edge $\{i,j\} \in E$  of the graph $G_\varphi = (V,E, (\varphi_{i,j}))$.
A modified version of \cref{eq:grstate} which describes the weighted graph state is

\begin{align}
    \ket{G_\varphi} = \negphantom{ii}\prod_{\{i,j\}\in E}\negphantom{ii} \CP_{i,j}(\varphi_{i,j}) \ket{+}^{\otimes \lvert V \rvert }, 
    \label{eq:weightgrstate}
\end{align}
where $\CP(\varphi)=\diag(1,1,1,\exp(i\varphi))$ on the $\{i,j\}$ subsystem, similarly to \cref{eq:czaction}. Such states arise naturally in systems with Ising-like interactions, and offer a useful representation of noisy graph state preparation using such a scheme. Unfortunately, the analogue of the procedure described by \cref{eq:stabgenorigin} does not produce a Pauli string, and hence a simple stabilizer description of correlations generally does not exist. For a standard graph state $\ket{G}$, all phases $\varphi_{i,j}$ above are equal to $\pi$.

\subsection{Prominent examples}

Many known families of states with structured entanglement can be represented as graph states. For instance,  the  $N$-qubit \emph{Greenberger–Horne–Zeilinger (GHZ) state} \cite{greenberger1989going}, defined as
\begin{align}
    \ket{GHZ} = \frac{1}{\sqrt{2}} ( \ket{0_1 0_2 \dots 0_N} + \ket{1_1 1_2 \dots 1_N} ),
\end{align}
is local unitary equivalent to the graph state of the fully connected $ N$-vertex graph or, equivalently, to the $ N$-vertex star graph.
In measurement-based quantum computing, the \emph{cluster states}, corresponding to grid graphs, are relevant \cite{Raussendorf2003measurement} and are an example of a broader category of universal resource states \cite{Van_den_Nest2006}. Intermediate-size cluster states have recently been realized on various platforms \cite{cao2023generation,Thomas2022}.
A detailed discussion of graph states, including weighted graph states, can be found in \cite{hein2006entanglement,vandenNest04graphicalClifforts}.

\section{Correlations and measurement} \label{sec:extraction}

In this section we show how local measurement affects correlations present in graph states. In particular, it is possible to extract an entangled pair of qubits $(i,j)$ from a graph state $\ket{G}$ by means of local Pauli measurements, if there exists a path between $i$ and $j$ through the graph $G$  \cite{hahn2019quantum,Raussendorf2003measurement}. As we later show in \cref {sec:noise_extract}, the basic procedure outlined in \cref{fig:fiveqbpath} is susceptible to noise effects, and the general mathematical language introduced here is used in the analysis and proposal of alternative extraction procedures.

The post-measurement state after applying a sequence of Pauli measurements can be determined by applying graphical rules for measurements \cite{hein2006entanglement}. In this section, we present
an equivalent 
method using stabilizer operators to find the post-measurement states. This method is better suited for calculations.
In general, there are multiple choices of measurement sequences. If noise is present, the entanglement quality of the resulting two-qubit state depends on the choice of the sequence.

Consider a graph state $\ket{G}$ that undergoes a local sequential measurement process involving Pauli operators.  
The state of the unmeasured qubits is completely characterized by the measurement pattern, that is, the choice of sequential Pauli measurements, along with the outcomes. The remaining correlations stem from the stabilizer operators of $\ket{G}$ consistent with the measurement pattern.

To see the structure of these correlations, let us choose a measurement pattern on a subset of qubits $I$: for each qubit $i\in I$, a local Pauli $o_i\in \lbrace X_i,Y_i,Z_i \rbrace$ is chosen. Then, each qubit $i$ is measured sequentially in the eigenbasis of the associated operator $o_i$. As mentioned previously, the expectation value of a Pauli string $O=\prod_{i\in I} o_i$ can be determined: if $O\in S_G$, it implies $\expval{O}_G = 1$. 
The case of $-O\in S_G$ differs only by the sign, and if neither $\pm O$ belongs to $S_G$, the expectation value is 0.

Without loss of generality, 
let us concentrate on the positive sign case and consider a preparation of $\ket{G}$ and subsequent measurement according to $\{o_i\}_{i\in V}$: the measurement of $o_i$ on the $i$-th qubit has an outcome of $s_i = \pm 1$. If $\langle O\rangle_G = 1$, the outcomes of the sequential measurement must reflect that: the signs of the measurement results must multiply to $+1$. Any other outcome pattern would decrease the magnitude of the expectation value (defined empirically as average over a random sequence of results) and is therefore 
inconsistent with the graph state correlations.
The condition  $ \prod_{i \in I} s_i = 1$ on the outcomes implies that not every combination of measurement outcomes $\{ s_i \}_{i \in I}$ can be measured\footnote{When saying that measurement outcomes are `consistent' with the stabilizer structure, we mean that this condition is fulfilled. That is the product of the outcomes multiply to the expected value.}. Note that for $\expval{O}_G = 0$, all combinations of outcomes $s_i$ are possible, so that both options $ \prod_{i \in I} s_i = \pm 1$ must appear.

This reasoning applies as well if the sequential measurement is stopped at any point and resumed afterward. If $O=\prod_{i\in I} o_i$ is a stabilizer operator, and only the qubits $I'\subset I$ were measured, the %\sout{sign structure of the further measurement of} 
correlations of outcomes of further measurement in  $I\setminus I'$ can be predicted. Thus, by the action of partial measurement, correlation in the remaining qubits is induced. 

The following Lemma captures this line of thought: the end correlations are defined by stabilizer operators consistent with the measurement scheme. Note that the measurement scheme does not itself have to be a stabilizer operator: only part of it has to extend to one, and such an extension always exists (see \cref{lem:postfuldef} in \cref{app:postmeascorr}).

\begin{lemma}
Let $\ket{G}$ be a graph state determined by the graph $G=(V,E)$. We denote the stabilizer of $\ket{G}$ by $S$. If a subset $I$ of qubits is measured in such a way that a qubit $i\in I$ is measured in the eigenbasis of $o_i \in \{X_i, Y_i, Z_i\}$, we encode this measurement scheme as the Pauli string $O=\prod_{i\in I} o_i$. 
The qubits of $I$ are measured independently and sequentially, with the measurement outcome of $o_i$ denoted by $s_i=\pm 1$.

Let us now write a stabilizer operator $Q\in S$ of $\ket{G}$ as a Pauli string with support on indices $J$: $Q=\pm \prod_{j\in J} q_j$, where $q_j$ is a Pauli operator in $\{X_j, Y_j, Z_j\}$, acting on qubit $j\in J$. 
The post-measurement state $\ket{\psi}$  is determined by 
those elements of the stabilizer, which are
consistent with the measurement pattern $O$, that is, 
$Q \vert_{I \cap J} = O\vert_{I \cap J}$, and relevant measurement outcomes $\{s_i\}.$ 
The stabilizer operators of the post-measurement state $\ket{\psi}$ can be computed from such operators $Q$ as
\begin{equation}
Q':=\pm \prod_{i\in I\cap J} s_i \prod_{j\in J\setminus I} q_j,
\label{eq:postmeascorrstab}
\end{equation}
where the global sign $\pm$ in front of the product is equal to the one appearing in the definition of $Q$.
\label{lem:lem1}
\end{lemma}
\begin{proof}
Consider the stabilizer operator $Q$ defined as stated in the Lemma. Then assume that the measurement outcomes of each of the qubits $i\in I\cap J$ is $s_i$. The projective measurement operators of these qubits can be written as $\frac12\left(\mathbbm{1}_i + s_i q_i\right)$. For instance, the projection operator associated with the $X+$ outcome is 
$
\dyad{X+}
=\frac12 \left(\mathbbm{1}+X\right)$. 
Since the projection operators commute with $q_j$ themselves, the following chain of equalities holds:
\begin{align}
    \ket{\psi}&=C \prod_{i\in I\cap J} \frac{\mathbbm{1}_i + s_i q_i}{2}\!\!\!\underbrace{\ket{G}}_{~~=Q\ket{G}} \notag \\
    &= C \prod_{i\in I\cap J} \frac{\mathbbm{1}_i + s_i q_i}{2} \left(\pm\prod_{j\in J} q_j\ket{G}\right) \notag \\
    &= \pm \!\!\!\prod_{j\in J\setminus I} \!\!\!q_j  \prod_{i\in I\cap J}\!\!\! s_i\underbrace{\left(C \!\!\!\prod_{i\in I\cap J} \!\!\!\frac{\mathbbm{1}_i + s_i q_i}{2} \ket{G}\right)}_{~~=\ket{\psi}} \notag\\
    &= \underbrace{\left(\pm\!\!\! \prod_{i\in I\cap J} s_i \prod_{j\in J\setminus I}\!\!\! q_j  \right)}_{~~=Q'} \ket{\psi}.
    \label{eq:lem1proof}
\end{align}
Here, $C$ is a normalization constant, and the following property
\begin{equation}
    \frac{\mathbbm{1} + s q}2 q = s \frac{s q + \mathbbm{1}}2 \quad \text{ for } s = \pm 1    
\end{equation}
was used, valid for any $q$ such that $q^2=\mathbbm{1}$ and in particular for $q\in\{X,Y,Z\}$.
\end{proof}

\begin{figure}%[t]
    \centering
    \includegraphics[width=\linewidth]{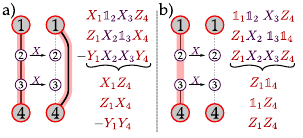}
    \caption{Partial local measurements on a graph state induce correlations in the remaining qubits.
	a) If the 4-vertex path graph state is measured according to the shown measurement pattern, that is, measuring the inner vertices in the $X$-basis, we get a connected two-vertex graph.
	b) The graph here is not connected: the edge $\{2,3\}$ is absent, possibly due to the scenario outlined in \cref{sec:uenoise}. If the corresponding graph state is measured according to the same measurement pattern, we get a disconnected two-vertex graph. Measurements can never induce entanglement between parties that were separable before.  
	}
    \label{fig:fiveqbpath}
\end{figure}

As an example, consider the path graph of 4 qubits, as shown in \cref{fig:fiveqbpath}a and the Pauli string $O = X_2 X_3$ which has support on the vertices in $I = \lbrace 2, 3 \rbrace$. The stabilizer of the four-qubit path graph contains $2^4$ stabilizer elements. Three of them are $Q^{(a)} = X_1 X_3 Z_4$, $Q^{(b)} = Z_1 X_2 X_4$, and $Q^{(c)} = - Y_1 X_2 X_3 Y_4$ which have support on qubits in $J^{(a)} = \lbrace 1, 3, 4 \rbrace$, $J^{(b)} = \lbrace 1, 2, 4 \rbrace$, and $J^{(c)} = \lbrace 1, 2, 3, 4 \rbrace$, respectively. We have $Q^{(a)} \vert_{3} = O\vert_{3}$, $Q^{(b)} \vert_{2} = O\vert_{2}$, and $Q^{(c)} \vert_{2,3} = O\vert_{2,3}$.\footnote{The notation $O|_{S}$ denotes restriction of a product of operators to a qubit subset $S$.}  Following \cref{lem:lem1}, the post-measurement state after measuring qubits 2 and 3 in the $X$-basis is stabilized by  $s_3 X_1 Z_4$,  $s_2 Z_1  X_4$, and $- s_2 s_3  Y_1  Y_4$. 
For $s_2, s_3 = +1$, these are stabilizer operators of the two-vertex graph state with one edge between them, as shown in the lower part of \cref{fig:fiveqbpath}a. For other measurement outcomes $s_2, s_3$, we get stabilizer operators which are equivalent up to local unitary operations. 
 Recall that this graph state is locally unitarily equivalent to the Bell states.

Another example of a four-vertex disconnected graph is shown in \cref{fig:fiveqbpath}b. 
This graph state is also
stabilized by the $2^4$ operators. The three of them that coincide with the measurement pattern $O = X_2 X_3$ are $Q^{(a)} = Z_1 X_2 $, $Q^{(b)} = X_3 Z_4$, and $Q^{(c)} = Z_1 X_2 X_3 Z_4$ that have support in $J^{(a)} = \lbrace 1, 2 \rbrace$, $J^{(b)} = \lbrace 3, 4 \rbrace$, and $J^{(c)} = \lbrace 1, 2, 3, 4 \rbrace$, respectively. We have $Q^{(a)} \vert_{2} = O\vert_{2}$, $Q^{(b)} \vert_{3} = O\vert_{3}$, and $Q^{(c)} \vert_{2,3} = O\vert_{2,3}$.  Following \cref{lem:lem1}, the post-measurement state after measuring qubits 2 and 3 in the $X$-basis is stabilized by  $s_2 Z_1 \1_4$,  $s_3 \1_1  Z_4$, and $s_2 s_3  Z_1  Z_4$. No correlations are present, since both remaining qubits are stabilized by independent local operators, and the resulting state is a product one.

Note that the stabilizer operators of post-measurement states are not always graph state stabilizer operators, but are always locally unitarily equivalent\footnote{More precisely, the post-measurement states are stabilizer states \cite{hein2006entanglement} and all stabilizer states are local Clifford equivalent to graph states \cite{vandenNest04graphicalClifforts}. The set of local Cliffords is a subset of the set of local unitaries.} 
to them \cite{hein2006entanglement}. 
The example in \cref{fig:fiveqbpath}a is chosen so that the post-measurement state does not require these local unitary transformations. Choosing a path graph of odd length results in different Bell state stabilizers.  
In the example shown in \cref{fig:fiveqbpath}b, we get compositions of Pauli $Z$-operators as post-measurement stabilizers. Applying the local unitary $H_1 H_4$ (where $H$ is the Hadamard gate) to them leads to the actual graph state stabilizer operators.

The remaining qubits are fully defined by the correlations developed during the measurement process. Thus, it is possible to find measurement schemes for which the output state has certain properties. For example, if the goal is to extract a Bell pair from a larger graph state $\ket{G}$, a measurement scheme $\prod_{i\in I}o_i$ must be found such that there exist stabilizers $Q$ of $\ket{G}$ such that $Q_{I\cap J}=O_{I\cap J}$ and ensuring Bell-like correlations between the terminal vertices. For example, in the 4-qubit path graph of \cref{fig:fiveqbpath}a and the above chosen measurements, the Bell state stabilizers are extracted. More elaborate measurement patterns have been studied for use in quantum networks \cite{hahn2019quantum,Mannalath2023,mannalath2022entanglement}. 
 Further examples for graph families analyzed later in this article can be found in \cref{app:postmeascorr}.

As a direct result of the  \cref{lem:lem1} we get the following observation:
\begin{corollary}
If a stabilizer operator $Q=\pm \prod_{j\in J} q_j$ is fully embedded in the measurement scheme 
$O=\prod_{i\in I} o_i$ such that $J\subset I$ and $Q=\pm O|_J$, the only observable 
combinations of measurement outcomes within $J$ are those such that $\prod_{j\in J}s_j = \pm 1$, where the sign is determined by the sign $\pm$ in front of the product defining $Q$.
\label{cor:obsstruct}
\end{corollary}
The proof follows by considering only $J\subset I$ in \cref{lem:lem1}.
As an example, consider the graph state shown in \cref{fig:fourqubit}a, with the measurement scheme $O=X_1 X_3$. This operator stabilizes the relevant graph state, and in the language of \cref{cor:obsstruct}, $Q=O$. The allowed measurement outcomes associated with this setup are $(+1,+1)$ and $(-1,-1)$. If we measure outcomes with different signs, we conclude that we had a different initial state.  Later we will use this observation to detect noise in states.

\section{Imperfect graph states} \label{sec:noise}
In this section, we analyze three classes of physically motivated noise models. They all stem from different methods of graph state preparation. The origins and effects are only sketched here due to complexity; please refer to \cref{app:noise} for the detailed derivation.

As outlined in the previous section, for perfect graph states Bell pairs can be extracted, provided the vertices of interest are connected by a path. Once a path is found, vertices adjacent to it are measured in the $Z$-basis and then the path vertices in the $X$-basis. In  \cref{sec:noise_extract}  we show that in the presence of noise, the performance of a simple path graph is not optimal, and the choice of other graphs can enhance the entanglement quality of the remaining qubits.

As a minimal example, consider the graph shown in \cref{fig:fiveqbpath}b. It may arise during preparation of the path graph (\cref{fig:fiveqbpath}a) via an imperfect $\CZ$ gate, and in such a case the Bell pair can not be generated: the two graph parts are uncorrelated, and local measurements can not produce entanglement from a product state. In a path graph, even a single missing edge leads to this scenario.

Regardless of origin, noise processes often have a stabilizer-consistent description. They produce a mixture of stabilizer states, or the noise effect can be otherwise understood using the stabilizer-like correlations. In \cref{sec:noise_extract}, we show how to use this observation to mitigate the effects of noise.

Implementations like superconducting qubit quantum computers \cite{Stehlik2021,long2021universal}, or ion traps \cite{wunderlich2009two,piltz2014trapped} prepare graph states by letting an initial product state evolve via an engineered Ising-like interaction Hamiltonian
\begin{align}
    H = \sum_{i\neq j} \gamma_{i,j} \frac{(Z_i+\1) \otimes (Z_j +\1)}4,
    \label{eq:isingham}
\end{align}
where $\gamma_{i,j}$ is the interaction strength between qubits $i$ and $j$.
Evolution according to such an interaction pattern effectively implements a sequence of controlled-phase operators $\CP_{i,j}=\diag(1,1,1,\exp(-i\varphi_{i,j}))$. The structure of phases $\varphi_{i,j}=\gamma_{i,j} t$ is determined by the interaction strengths $\gamma_{i,j}$ and the evolution time $t$. Inaccuracies in controlling either of them lead to the generation of a weighted graph state \cite{Hartmann_2007}. However, under general assumptions, the effective quantum channel implemented this way has a simple decomposition into local (one or two edges) operations, on top of the product of a sequence of $\CZ$ operations preparing the desired state. 

Ideally, all of the generated phases $\varphi_{i,j}$ are equal to $\pi$. More realistically, with each experimental run, the actual realization will be $\varphi_{i,j}=\pi+\varepsilon_{i,j}$ with $\varepsilon_{i,j}\ll 1$. We discuss the two extreme cases of correlations of phase noises $\varepsilon_{i,j}$. The first, \emph{uncorrelated phase noise}, assumes that the noises at different edges $\{i,j\}$ are completely independent. The other, \emph{correlated phase noise}, assumes the contrary: each phase noise factor is equal to any other in a single experimental run, $\varepsilon_{i,j}=\varepsilon$. The noise encountered in experiments likely has yet a different structure with partial correlations between the edges; still, the two extremal processes may help in modeling the real-life scenario, and strategies developed to mitigate them should apply in the more general cases.

\subsection{Uncorrelated edge noise} \label{sec:uenoise}

The effect of phase noise is especially simple if the phases $\varphi$ are distributed symmetrically and uncorrelated around $\pi$. The resulting channel is described by a probabilistic application of $\CZ$:  with probability $(1-p)$, an edge is created via $\CZ$, and with probability $p$, no operation is performed. See \cref{app:uncorr_phase_noise} for a more detailed discussion.
Each of the controlled-$Z$ unitary operators in  \cref{eq:grstate} is replaced with the quantum channel
\begin{align}
    \Cp[\rho] = (1-p)\Cz[\rho] + p \rho,
\end{align}
where $p$ depends on the distribution of $\varepsilon$ (see \cref{eq:pasphase}), $\Cz[\rho]=\CZ \rho \CZ^\dagger$, and $\rho = \dyad{+}^{\otimes \abs{V}}$. 
Thus the effect of noise consists in the generation of \emph{randomized graph states} \cite{Wu_2014_random_graph}: an ensemble of graph states built atop the original graph $G$ by removing the edge $\{i,j\}\in E$ with probability $p_{i,j}$. Related graph ensembles have been introduced in \cite{Erd_s2022} in a classical context. For the constant probability, the end state can be written as
\begin{equation}
    \rho = \sum_{E' \subset E} (1-p)^{\vert E' \vert} p^{\vert E \setminus E' \vert} \dyad{G'},
    \label{eq:randomgraph}
\end{equation}
where $G' = (V,E')$ is a graph with edge set $E'$ being a subset of the edge set $E$ of the original graph $G$. 

Note that it is an \emph{effective description}: in each experimental run, the prepared state is a pure  \emph{weighted graph state} \cite{Frantzeskakis23weightedgraphs}. The density operator in \cref{eq:randomgraph} is a state of knowledge about the system, averaged over different noise realizations.

\subsection{Correlated edge noise} \label{sec:cenoise}
If the weights in the Hamiltonian of \cref{eq:isingham} are all equal ($\gamma_{i,j} = 1$) but the interaction time is not perfectly controlled, within one experimental run, all the resulting phases $\varphi_{i,j}$ are equal.   This is the \emph{correlated phase noise}: it has an effect similar, but more involved, to the uncorrelated case. Here, we also assume that $\varphi$ is symmetrically distributed around $\pi$. 
Let us assume that in each experimental run $\varphi_{i,j}=\varphi=\pi+\varepsilon$, and $\varepsilon$ is distributed with the normal distribution of the standard deviation $\sigma$. The resulting state can be described by a parameter $p$ depending on the distribution of $\varepsilon$ (\cref{eq:pasphase}). By keeping only the lowest nontrivial order of noise effects and subsequent averaging, the end state can be described with the help of square roots of the unitary operator $\CZ$:
\begin{align}
\rho =& (1-p \abs{ E} )\dyad{G} 
+ p\sum_{G' \subseteq G } \dyad{G'} \\
&+ \frac{p}{2} \sum_{\substack{e, e'  \in E \\ e \neq e'}} \sum_{s, s' \in \{\pm 1\}} s s' \CS_e^{s} \CS_{e'}^{s'}   [\dyad{G}], \notag
\end{align}
where the summation runs over graphs $G' \subseteq G $ with single edges removed from $G$, and $\CS^\pm[\rho]$ is the application of one of the two unitary square roots of $\CZ$. 
See \cref{app:noise} for details.

\subsection{Local $Z$ flip noise} \label{sec:lznoise}
In linear optic experiments, where Bell pairs are generated and subsequently fused to get a graph state, the entangling operation may fail as a result of partial photon distinguishability. 
This can be caused by imperfect frequency or spatial mode overlap \cite{Browne2005, Rohde2006, Rahimi_Keshari2016}.
 See  \cref{app:corr_phase_noise} for a derivation. The effect of a noisy fusion gate described this way can be modeled as a \emph{perfect} fusion followed by a probabilistic application of the $Z$ unitary to the surviving optical qubit $i$:
\begin{align}
    F_i (\rho) = (1-p) \rho + p Z_i \rho Z_i.
\end{align}
Here, the probability $p$ depends on the level of photon distinguishability of the fused qubits. Thus, the final graph state developed in this procedure can be thought of as a probabilistic application of local $Z$ unitaries to each of the qubits in the graph.

\subsection{Quantification of the noise effects} \label{sec:quannoise}

In the following sections, we determine the entanglement quality of two- and three-qubit Bell and GHZ graph states left after a measurement procedure performed on a graph state. 
To quantify the entanglement quality of the post-measurement state, we use the fidelity with respect to the ideal state $\ket{\psi}$.
If the post-measurement state is $\sigma$, the following holds:
\begin{align}
    \mathcal{F}(\ket\psi\!\bra\psi,\sigma)=\langle \psi \vert \sigma \vert \psi \rangle.
    \label{eq:fidelityone}
\end{align}
Note that the ideal post-measurement state $\ket{\psi}$ depends on the observed outcomes during the experiment: they do affect the signs in the final stabilizer structure (see \cref{lem:lem1}). Therefore, the ideal state $\ket\psi$ must be determined from the signs of the outcomes, provided they are consistent with the stabilizer structure. 
 The different outcomes may appear with varying probabilities, and the expectation value of associated fidelities (\cref{eq:fidelityone}) is the average fidelity $\langle\mathcal{F}\rangle$. This is effectively approximated by taking $N$ samples from the defining ensembles for a specific noise model with a Monte Carlo algorithm and thus
\begin{align}
    \langle\mathcal{F}\rangle \approx \frac1N\!\!\! \sum_{\text{outcomes}}\!\!\! \mathcal{F}(\text{outcome}).
    \label{eq:avfidelity}
\end{align}

Since we are mostly interested in the behavior in the low noise limit, we define \emph{fidelity susceptibility} $\alpha$ as the rate of change of mean fidelity $\langle\mathcal{F}\rangle(p)$:  the ensemble of extracted and postselected states depends on the noise, and the susceptibility captures the effects of small but not negligible noise:
\begin{equation}
\alpha = - \left.\frac{\dd \langle \mathcal{F} \rangle(p)}{\dd p}\right|_{p=0}.
\label{eq:fidsus}
\end{equation}

For zero noise, the fidelity is equal to 1 by definition. For any other amount, the fidelity may only decrease, and hence the susceptibility is positive. Its numerical magnitude determines how fragile the extraction procedure is to the effects of noise: a robust one has a small 
susceptibility.

Practically, the fidelity susceptibility is estimated here by a discrete derivative approximation, as $\alpha\approx\frac{1-\langle\mathcal{F}\rangle(p_*)}{p_*}$. In our computations, we chose $p_* \approx 10^{-2}$. The noise value $p_*$ was chosen such that it is empirically low enough for reasonable approximation and high enough such that numerical noise is not relevant.

\section{Extraction from noisy ensembles} \label{sec:noise_extract}

In this section, we apply the methods developed in previous parts to show the effects of the noise and find measurement patterns that minimize them. 
As mentioned previously, we analyze a task related to teleportation: the extraction of a high-quality Bell pair  from a larger noisy state, close to a graph state $\ket{G}$. We summarize our procedure as a Monte Carlo simulation of multiple experimental runs:
\begin{enumerate}[label=(\alph*)]
    \item We prepare a potentially noisy version of a graph state associated with one of the graphs shown in \cref{fig:families}. The analyzed noise families are described in \cref{sec:noise}.
    \item We measure all qubits associated with the inner vertices in the Pauli $X$ basis, see \cref{fig:crazy-structure}, obtaining random outcomes in each run.
    \item We postselect on the measurement outcomes: we keep the state if the product of measurement outcomes is equal to 1 and discard the state otherwise, in accordance with \cref{cor:obsstruct}.
    \item If we keep the state, the post-measurement state consists of the two qubits associated with the terminal vertices of the graph. The fidelity of the post-measurement state is evaluated with respect to the expected ideal case, see \cref{sec:quannoise}.
    \item For evaluating the noise resistance, we repeat the protocol multiple times so that the calculated average fidelity, given in \cref{eq:avfidelity}, converges. 
\end{enumerate}

The main results are shown in Figures \ref{fig:edgesus} and \ref{fig:corredgesus}.

We mostly concentrate on the scenario where the entangling gates are almost perfect, and in such a case the fidelity susceptibility ($\alpha$, defined in \cref{eq:fidsus}) quantifies the noise effects. 
For the estimates of $\alpha$  only a limited number of quantum states forming the ensemble (discussed in \cref{sec:noise}) is relevant, and for the calculation only this subset is taken into account, reducing the uncertainty inherent to the Monte Carlo procedure. The code used for calculation is publicly available \cite{zenodo}, and the numerical details can be found there.\\

\subsection{Noise-correcting structures and robust families of states}

As introduced in \cref{sec:noise}, noise processes lead to an ensemble of different states $\{(p_i, \ket{\psi_i} )\}$ being prepared, where one of the ensemble components is the ideal state $\ket{G}$. This ensemble can be interpreted as a probabilistic preparation of an unknown state $\ket{\psi_i}$ with probability $p_i$. The measurement procedure  might now yield results inconsistent with $\ket{G}$, and observations of those are an indication of a state other than $\ket{G}$ being prepared -- see \cref{cor:obsstruct}.

As an example, consider the square graph -- a $4$-cycle of vertices $\{1,2,3,4\}$, as shown in \cref{fig:fourqubit}a. The operator $O=X_1 X_3$ is an element of the stabilizer of $\ket{\square}$, thus the only 
outcomes that can be observed in the sequential measurement of qubits $1$ and $3$ are $+1,+1$ and $-1,-1$. Consider now all edge subgraphs $\{G'\}$ of the square graph and a physical process in which an ensemble of graph states $\{p_{G'}, \ket{G'}\}$ is prepared, some of them are shown in \cref{fig:fourqubit}b. 
The operator $O$ (and $-O$) is not a part of the stabilizer of any proper subgraph with at least one edge. Thus, 
measurement outcomes of $+1,-1$ or $-1,+1$ are possible. Observing those in the same measurement procedure implies that the pre-measurement state was not $\ket{\square}$. 
Therefore, measuring according to stabilizer operators and evaluating the measurement outcomes performs a probabilistic check whether $\ket{G}$ was truly prepared. We postselect accordingly, that is, we discard the state when the 
measurement outcomes do not agree with the expected ones. However, the existence of such parity-checking structures depends on the graph and the task for which the associated quantum state is used. More examples are discussed in \cref{app:postmeascorr}.

\begin{figure}
    \centering
    \includegraphics[width=.7\linewidth]{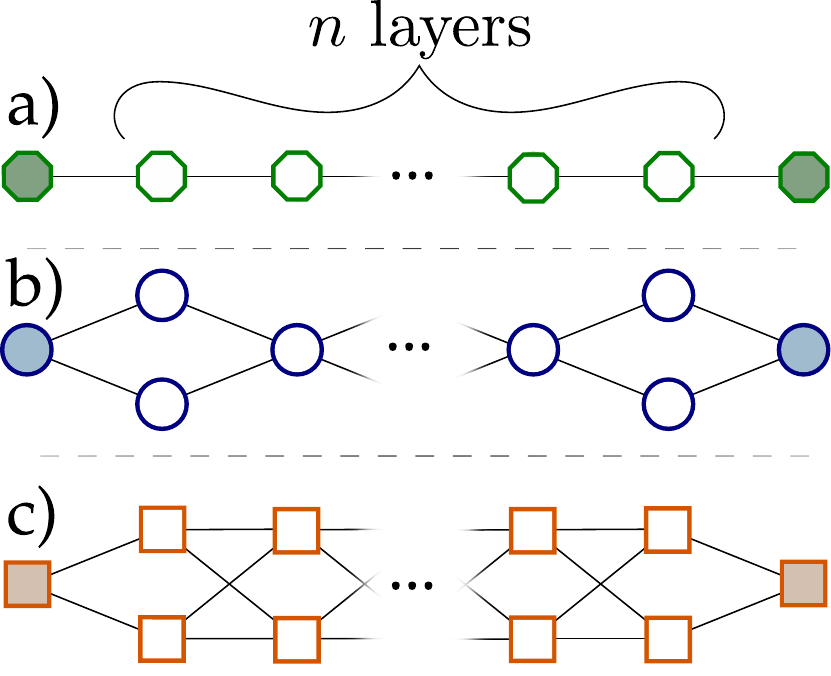}
    \caption{Considered families of graphs for the task of Bell pair extraction, parameterized by the length of the internal part $n$. The terminal qubits, designated for the creation of a Bell pair, are denoted with filled nodes. From the top: path graph, twisted pair, crazy graph. }
    \label{fig:families}
\end{figure}

A Bell pair can be prepared with the help of a path graph state and local $X$ measurements (\cref{fig:fiveqbpath}a), but such a measurement scheme does not allow for embedded stabilizer operator parity checks.
Other graphs and measurement patterns do better in this regard. 
 Here we present some alternatives. 

The families of graph states analyzed are parameterized by the length between the terminal qubits across which a Bell pair is to be prepared (see \cref{fig:families}). Apart from a simple path graph (\cref{fig:families}a), included as a benchmark for comparison, we include two other families described below -- both of them have their strengths.

The \emph{twisted pair} graph  \cite{anders2012twistedgraph}, shown in \cref{fig:families}b,  
is built from layers, where each qubit in the layer $k$ is connected with every qubit from the layers $k-1$ and $k+1$. The number of qubits in each layer alternates between $1$ and $2$, which enables the existence of a restricted set of stabilizers: for each 2-qubit layer consisting of the qubits $\{i,j\}$, the operators $X_i X_j$ stabilize $\ket{G}$. Additionally, it is very robust to certain types of noise because of its relatively simple geometry.

The \emph{crazy} graph similarly consists of layers of $2$ qubits each, with  full connectivity between the adjacent layers. It is known for its noise robustness  \cite{rudolph2017crazygraph,Morley-Short_2019_crazygraphs}, which stems from the structure of its stabilizer operators. See \cref{fig:crazy-structure} and \cref{app:postmeascorr} for more details. 

\begin{figure}[t]
    \centering
    \includegraphics[width=\linewidth]{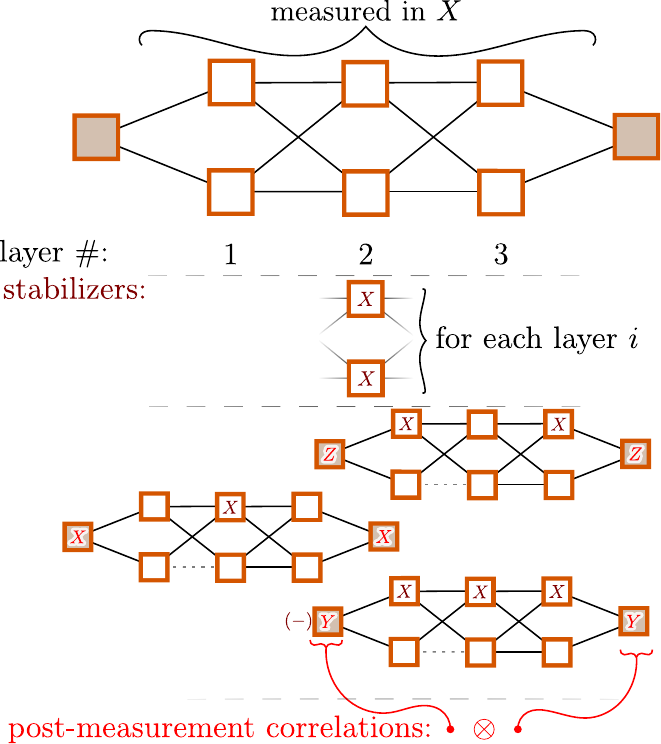}
    \caption{Crazy graph of length 3 together with an all-$X$ measurement pattern possesses embedded error-checking and correction structures. Middle: each layer is associated with an $X\otimes X$ stabilizer, which can be used to confirm the proper preparation of the state. Bottom: even if any single internal edge (dotted line between layers 1 and 2) is missing, the terminal qubits correlations are not affected: the operators shown are stabilizers, regardless of whether the edge is present or not.}
    \label{fig:crazy-structure}
\end{figure}

Each of the graph states associated with the aforementioned three families (path graph, twist graph, and crazy graph) can be used for the extraction of a Bell pair across the terminal qubits (shaded vertices in \cref{fig:families}). Only local $X$ measurements are used for all three, and in the ideal case, a maximally entangled state is produced.
 This is because there exist stabilizer operators consistent with this measurement pattern, as in the exemplary 4-qubit path graph discussed in \cref{sec:extraction}.

Similar structures can be found for all three graph families mentioned. We find sets of stabilizer operators 
which have $X$ operators in the internal section and Pauli operators which ensure correlations on the terminal qubits -- see \cref{app:patterns}. These operators are 
consistent with the measurement scheme.

\subsection{Bell pair quality scaling for different types of noise}

We analyze the approach of noise-detecting structures for the graph families introduced above.  The procedure is employed for different layer sizes $n$, see Figures \ref{fig:edgesus} and \ref{fig:corredgesus}. We find that the twisted pair and crazy structures are less affected by the noise, as evidenced by the fidelity susceptibility  $\alpha$ (\cref{eq:fidsus}). The value of $\alpha$ is smaller, and in the case of the crazy graph does not depend on the graph length. Furthermore, the susceptibility for the crazy graph structure does not depend on the length, as a result of the correctional stabilizers mentioned \cref{app:patterns} and shown in \cref{fig:edgesus}. This behavior is observed both for uncorrelated phase noise, resulting in probabilistic edge losses and a local $Z$ flip noise.

The results can be understood as the destruction of the perfect correlations present in graph states, mathematically described in \cref{sec:extraction}. Even if for every ensemble element the end state is maximally entangled, the mixture of such states might have less or no entanglement. Postselection on the measurement outcomes consistent with the ideal stabilizer structure helps twofold: one of its effects is probabilistic error detection. The other, appearing in the case of the crazy graph ensemble, is more subtle: after postselection, the end state may be the same for a large portion of the ensemble elements (see \cref{app:postmeascorr}).

\begin{figure}[t]
    \centering
    \includegraphics[width=\linewidth]{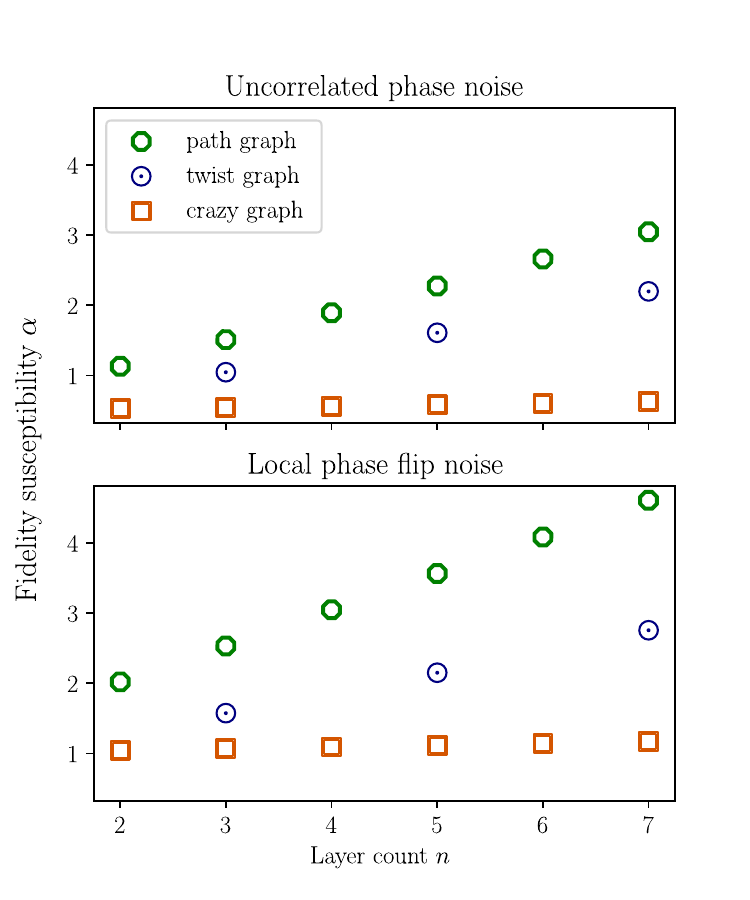}
    \caption{The effects of the edge loss (top) and local phase flip (bottom) noise on the task of Bell pair extraction for the three analyzed graph state families. The susceptibility $\alpha$ (\cref{eq:fidsus}) is the initial falloff of fidelity in the low noise limit of the resulting two-qubit state with respect to the expected Bell pair. The results here only take into account the measurement outcomes consistent with the ideal state. In this case (low noise limit) the dominant contributions are the single edge losses (top) and single qubit flips (bottom). The constant susceptibility of the crazy graph is a result of the noise at the terminal edges only: each internal edge loss is automatically corrected for, and every single flipped qubit is detected. 
    }
    \label{fig:edgesus}
\end{figure}

An additional entanglement-preserving structure can be found in the case of perfectly correlated phase noise. This is the source of zero susceptibility of the twist graph and periodic behavior of the crazy graph in \cref{fig:corredgesus}. The authors of Reference \cite{Frantzeskakis23weightedgraphs} 
recently observed that in the case of the path graph, postselection on different measurement outcomes yields vastly different results in 
terms of entanglement quality. We have found that this result does generalize to the other linear graphs analyzed by us. 
Both for the crazy and  twist graph families the lowest order nontrivial noise effects can be canceled completely if all the measurement outcomes are postselected on measuring the $(-1)$ outcome of $X$.

\begin{figure}[t]
    \centering
    \includegraphics[width=\linewidth]{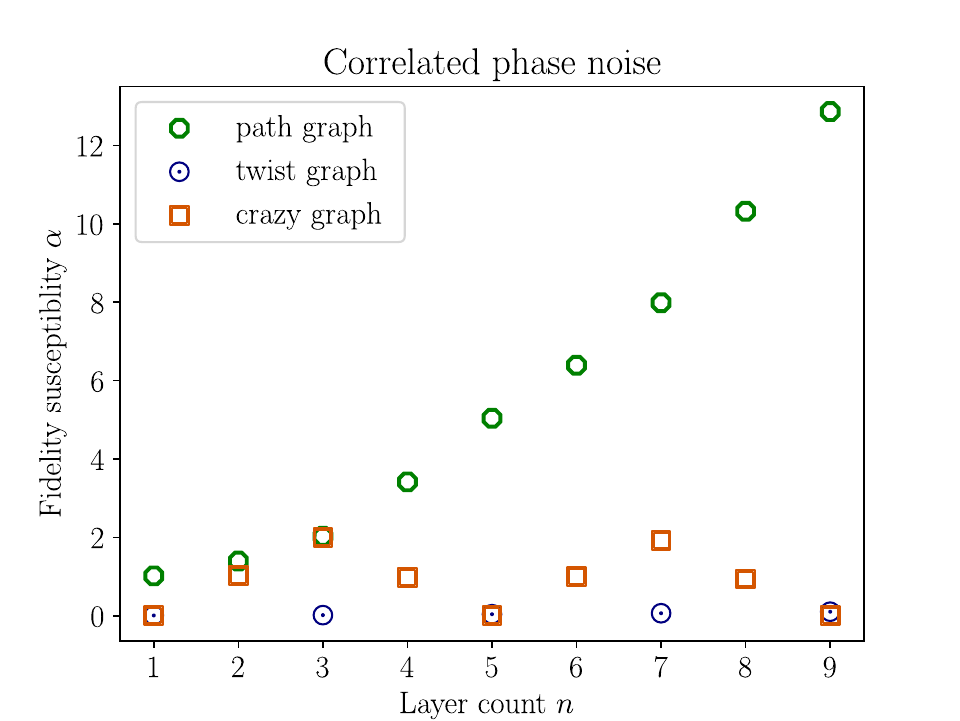}
    \caption{The effects of correlated phase noise on the task of Bell pair extraction for the three analyzed graph state families. The susceptibility $\alpha$ (\cref{eq:fidsus}) is the initial falloff of fidelity in the low noise limit of the resulting two-qubit state with respect to the expected Bell pair.   Here, only the cases postselected on observing exclusively minus signs during the $X$ measurement are taken into account. In this case (low noise limit with restrictive postselection) the twist graph is able to remove the effect of noise completely, and the crazy graph shows a periodic pattern of susceptibilities, still significantly reducing the noise effect. The situation inverts for observing exclusively $+$ signs: then, the noise gets amplified. 
    }

    \label{fig:corredgesus}
\end{figure}

This is consistent with our previous observations: such an observed measurement pattern does not violate parity constraints arising from the stabilizer structure. This result can be understood by studying the modifications of the stabilizer operators of $\ket{G}$ under correlated controlled-phase noise -- see \cref{app:corredgecancel}.

\section{Extraction of GHZ states and large noise regime} \label{sec:GHZ}

\begin{figure}
\centering
\raisebox{3cm}{a)}\includegraphics[width=.7\linewidth]{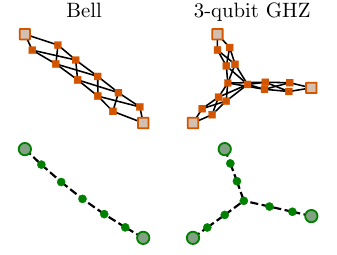}\\
\raisebox{2cm}{b)}\includegraphics[width=1\linewidth]{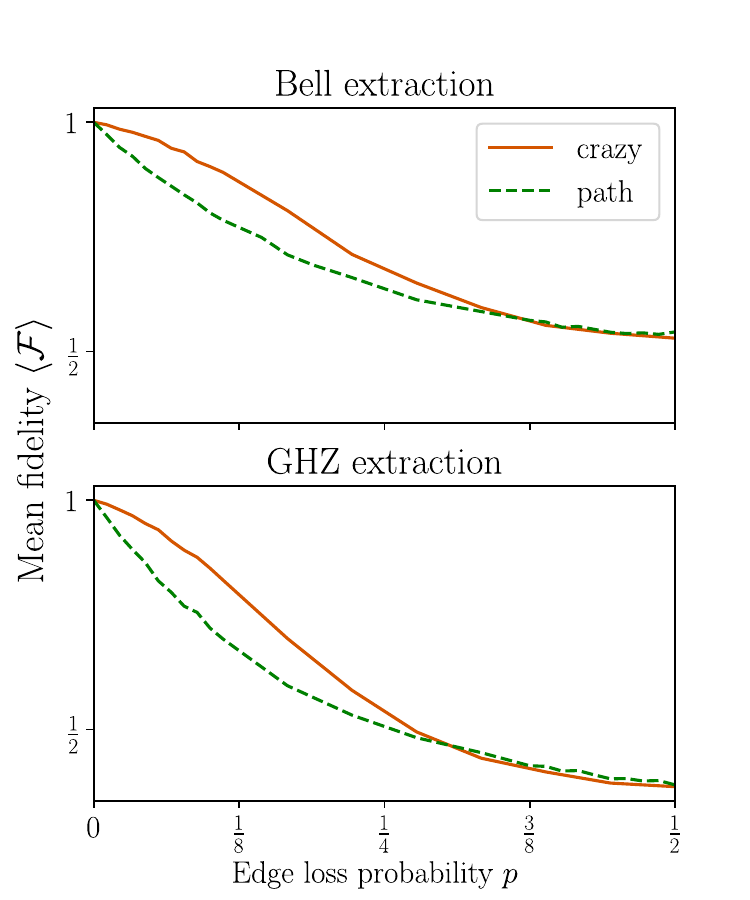}
\caption{ a) Structures of the analyzed graphs: path and crazy templates yielding Bell pairs and 3-qubit GHZ states upon measurements of the central qubits. Similarly to the previous case, all qubits except the terminal ones (highlighted and enlarged)  are measured in the $X$ basis. b) Fidelity of the end state with respect to the ideal Bell and GHZ states as a function of noise. The fidelity is averaged over possible measurement outcomes consistent with the measurement-embedded stabilizer structures. The initial fall-off rate of fidelity at $p=0$ is the \emph{fidelity susceptibility}, defined in \cref{eq:fidsus}.  
}
\label{fig:constlengthghzbellres}
\end{figure}

In \cref{sec:noise_extract} we concentrated on the Bell pair extraction for various geometries in the low noise regimes. However, the presented structures (twisted and crazy graphs) are of more general use. In this section, we analyze large noise effects on two exemplary scenarios. In addition to Bell pair extraction, we consider a similar procedure for GHZ states \cite{greenberger1989going} arising from a starlike graph, with analogous measurement scheme composed entirely of $X$ operators. The considered graph structures are presented in \cref{fig:constlengthghzbellres}a.  The main result here is the fidelity of output states (\cref{fig:constlengthghzbellres}b) postselected on observing the expected 
outcomes during $X$ measurement of the internal qubits.

The Monte Carlo simulations required to analyze the large noise regime are lengthy, as multiple ensemble elements are relevant for the result. Therefore, for simplicity we consider fixed graph geometries (path and crazy-like with fixed length) with only the uncorrelated phase noise present.

The fidelity $\mathcal{F}$ is determined as the average of ensemble elements and postselected stabilizer-consistent outcomes -- see \cref{eq:avfidelity}. This avoids direct calculations with density operators. The fidelity of the output state to the ideal pure state is estimated using this method until the plots converge smoothly and the variance of the fidelity estimator is negligible.

Analysis of the \cref{fig:constlengthghzbellres} shows that the fidelity response for both the GHZ and Bell state extractions exhibits qualitatively similar behavior. The output state is  more resilient to noise (as compared to the simpler path graph) also in this new task of extraction of the GHZ state. 
 Most importantly, the noise effect is greatly reduced in the intermediate noise limit.  This performance gain is due to the postselection involving noise-detecting and -correcting capabilities intrinsic to this graph family. Thus, the `crazy graph' structures can serve as a general link. They can generalize beyond Bell pair extraction to more complex multi-qubit systems, maintaining their noise-resistant properties.

 Further investigations of different connectivity structures suggest that this behavior is universal, although a formal proof is beyond the scope of this study. 
If true, the `crazy graph' can serve as a link similar to a simple path link. That is, it has similar functions but some level of error correction is implemented.

\section{Summary and Outlook} \label{sec:sum}

This study has advanced the understanding of noise processes in graph states, specifically addressing the effects of preparation by the Ising-like interaction and photonic qubit fusion. We have found graph structures applicable in error detection and correction, especially in tasks such as information transfer and GHZ state preparation. This approach uses additional qubits for probabilistic verification and demonstrates a new method to improve resilience against specific noise disturbances. 
 The mathematical formalism used in the calculations has been previously used in the nonlocality detection in graph states \cite{Meyer_2023_graph_inflation,meyer2024selftestinggraphstatespermitting} and could be extended to other families beyond Bell and GHZ states.

Future research may include more general tasks found in the measurement-based quantum computation approach. So far, our results are applicable to measurement patterns composed of Pauli operators, but it is known to restrict the set of possible operations to a classically simulable one. A similar method, based on non-Pauli measurements, but still allowing for noise effects reduction, could open new pathways for quantum computation, albeit with conceptual and practical complexities. The methods presented here can be generalized for arbitrary output graph geometry, and the applicability of this procedure for quantum metrology \cite{Shettell2020} could be studied.

Additional platforms for generating graph states can also be analyzed: in this context, sequential single-atom emitters \cite{Thomas2022} are especially interesting. Although noise processes are fairly complex, such systems are capable of producing extensive path graph states and present a unique opportunity for advancing graph state generation techniques. Different experimental realizations natively support different connectivity structures \cite{cao2023generation,Thomas2022,wunderlich2009two,Li2019,Lu2007}, and optimal utilization of a given experimental setting may require further development of the presented techniques.

\section*{Acknowledgments}
We thank Frederik Hahn, Mariami Gachechiladze, and Jan L.\ B\"onsel 
for discussions. This work was supported by 
the Deutsche Forschungsgemeinschaft  (DFG, German Research 
Foundation, project numbers 447948357 and 440958198), the 
Sino-German Center for Research Promotion (Project M-0294), 
the ERC (Consolidator Grant 683107/TempoQ), the German 
Ministry of Education and Research (Project QuKuK, BMBF Grant 
No. 16KIS1618K), the Stiftung der Deutschen Wirtschaft,  
the European Union’s Horizon 2020 research and
innovation programme under the Marie Skłodowska-Curie grant agreement (No 945422), and
the Austrian Science Fund (FWF) project quantA [10.55776/COE1]. 
 Funding from the projects DeQHOST APVV-22-0570,
QUAS VEGA 2/0164/25, and the Stefan Schwarz programme is gratefully acknowledged. 
For Open Access purposes, the authors have applied a CC BY public copyright license to any author accepted manuscript version arising from this. 
The authors acknowledge TU Wien Bibliothek for financial support through its Open Access Funding Programme.

\onecolumn

\newpage
\appendix

\section{Bell pair creation is equivalent to teleportation}
\label{app:bellteleportation}
As mentioned in the introduction, the creation of a Bell pair is equivalent to the task of information transfer (teleportation) across the graph. This is so because if the Bell pair creation in the graph $G$ is viewed in the context of a larger graph $G'$, the entire process, followed by measurement of one of the terminal qubits, performs \emph{qubit fusion}: the unmeasured terminal qubit acts as the two combined. Since all local measurements with nonoverlapping supports do commute, this can be done before any other measurement or afterwards: in the latter case, the information created in one of the qubits is merged with the second one. This is captured by the following Lemma, represented pictorially by the \cref{fig:telep}. 
\newcommand{\inn}{{\operatorname{in}}}
\newcommand{\outt}{{\operatorname{out}}}
\begin{figure}[t]
\centering\includegraphics[width=.9\linewidth]{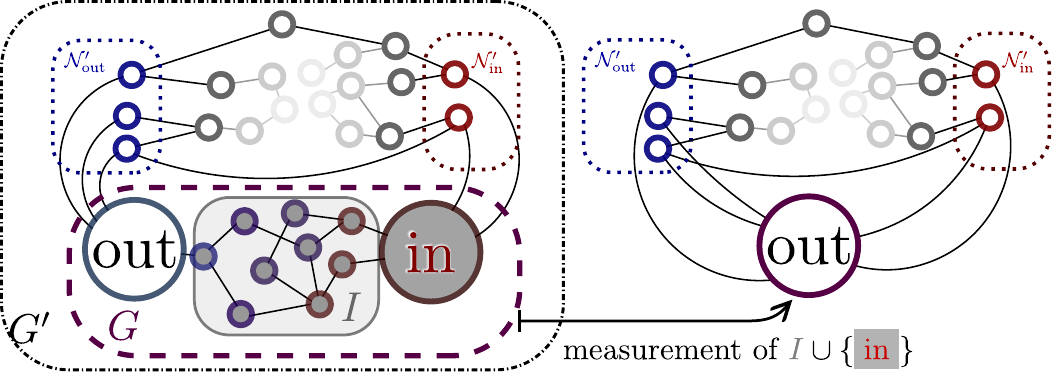}
\caption{If a measurement scheme generates a Bell pair across the vertices $\{\inn,\outt\}$ in the state corresponding to a graph $G$, in a larger $G'$ it can be used to transfer information across the graph. Extending the measurement by a properly chosen operator $A_{\inn}$ leads to the fusion of the two terminal vertices.}
\label{fig:telep}
\end{figure}

\begin{lemma}
Let us assume the following:
\begin{itemize}
\item We fix a graph $G=\{V,E\}$, and choose the two terminal vertices $\{\inn,\outt\}$ such that $V=\{\inn,\outt\}\cup I$.
\item The corresponding graph state $\ket{G}$ undergoes a sequential measurement of $I$ producing a Bell pair across $\{\inn, \outt\}$. According to the \cref{lem:lem1}, this means that the correlations between $\{\inn, \outt\}$ are described by stabilizer operators of $\ket{G}$ modified by the measurement process (\cref{eq:postmeascorrstab}).
\item Let $G$ be a part of a larger graph $G'$ in the following sense: $G'=(V',E')$ embeds $G$ in such a way that $\{\inn, \outt\}$ is a separating set. This means that the only connections between $G$ and the rest of $G'$ are at these two vertices.
\item Let us denote the vertices in $V'\setminus V$ to which both qubits are connected by $\mathcal{N}'(\inn)$ and $\mathcal{N}'(\outt)$, respectively, and assume that the two sets are disjoint.

\end{itemize}

Consider now the following process viewed from the perspective of the graph $G'$: 
\begin{enumerate}[label=(\alph*)]
    \item the measurement scheme, which within $G$ would produce a Bell pair across the terminal qubits, followed by
    \item a product of single qubit unitary operations on the qubits $\{\inn, \outt\}$, aimed to counteract the effect of measurement and bring back the stabilizer operators into a manageable form, and
   \item the measurement of $X_\inn$.    
\end{enumerate}
These operations taken together effectively perform qubit fusion of the terminal qubits. The new state is a graph state where the qubits $V'\setminus V$ are intact, and the new neighborhood of the remaining qubit is $\mathcal{N}_{\outt} \cup \mathcal{N}_{\inn}$.
\end{lemma}
\begin{proof}

If within $G$ a Bell pair is produced across the terminal vertices $\{\inn, \outt\}$ as a result of the measurement pattern $O=\prod_{i\in I} o_i$, this means that there exist three $O$-consistent stabilizers involving the terminal vertices in the sense of \cref{lem:lem1}. Combinatorial considerations show that one of those must be of a form such that the operators at the terminal vertices are $A_\inn, B_\outt \in \{X,Y\}$ (but not $Z$), chosen independently. 

Thus, there exists an $G$-stabilizer operator of the form $\pm A_\inn B_\outt \prod_{j\in J} o_j  = g_\inn g_\outt \prod_{j \in J^*} g_j$, with $J, J^* \subset I$. Note that the definition of $g_i$ depends on the graph in question (\cref{eq:generator}): if by $g'_i$ we denote the generators corresponding to the entire graph $G'$, the result is also a stabilizer operator, differing only by the $Z$ operators in the set $\mathcal{N}'(\inn) \cup \mathcal{N}'(\outt)$.

Let us now determine the post-measurement stabilizer operator within $G'$ by application of \cref{lem:lem1}. As prescribed by the operator $O$, every qubit $i\in L$ is measured in the basis of $o_i$, and the qubit $(\inn)$ is measured in the basis of  $A_\inn$. As a result, the post-measurement stabilizer operator is
\begin{equation}
    \left(\pm s_\inn \prod_{j\in J} s_j\right) B_\outt\negphantom{ii}\prod_{j' \in \mathcal{N}'_\inn \cup \mathcal{N}'_\outt}\negphantom{ii} Z_{j'} ,
\end{equation}
which can be taken to be a new generator associated with the vertex $\outt$, and brought to the canonical form of $X_\outt \prod_{j'\in\mathcal{N}'_\inn \cup \mathcal{N}'_\outt} Z_{j'}$ by local unitary basis change of the $\outt$-qubit. The properties of other stabilizers involving the qubits $\{\inn,\outt\}$ can be proven similarly. They merge to stabilizer operators of only the qubit remaining out-qubit.

Thus, after the measurement within the Bell-generating part $I\cup\{\inn\}$, the $\outt$-qubit behaves exactly like it was connected to the neighborhood $\mathcal{N}'_\inn$ as well: the two vertices are fused and the information between them is transferred. 
\end{proof}

\section{Post-measurement correlations and stabilizers %\sout{compatible} 
consistent with the measurement scheme}

\label{app:postmeascorr}
\begin{lemma}
    Fix a graph state $\ket{G}$ of $N$ qubits and a certain measurement scheme defined in \cref{lem:lem1} by a set of qubits $I$ and operators $o_i$. The stabilizer operators  of $\ket{G}$ which are 
    consistent with the measurement pattern (that is, the operators share elements with the measurement scheme, see discussion before \cref{eq:postmeascorrstab}) form a subgroup. This subgroup has $N-\lvert I\rvert$ generators, and therefore fully defines the post-measurement state.
    \label{lem:postfuldef}
\end{lemma}
\begin{proof}
The stabilizer of $\ket{G}$, up to the structure of the signs, can be viewed as a linear subspace of $\mathbb{Z}_2^{2N}$ -- this is known as the \emph{binary representation} of the stabilizer \cite{hein2006entanglement}. The identification is following: if the stabilizer operator is defined as $\prod_{j\in J} g_j$, the corresponding element of $\mathbb{Z}_2^{2N}$ is $(x_1, \ldots, x_N) \oplus (z_1, \ldots, z_N)$, with $x_i$ being 1 if and only if $i\in J$ and $\vec z = A_G \vec x \mod 2$, where $A_G$ is the adjacency matrix of $G$. Thus, if a stabilizer operator contains $X_i$, the corresponding numbers are $(x_i, z_i)=(1,0)$, for $Z_i$ it is $(x_i, z_i)=(0,1)$, and $Y_i$ corresponds to $(x_i, z_i)=(1,1)$; no operator at site $i$ is denoted by $(x_i, z_i)=(0,0)$. In this way, the additive algebra of $\mathbb{Z}_2^2$ mimics the product rules of Pauli operators with only the sign %\sout{structure} 
missing.

The stabilizer of $\ket{G}$ corresponds to a special $N$-dimensional subspace in $\mathbb{Z}_2^{2N}$ \cite{hein2006entanglement}. For any stabilizer operator $Q=\prod_{j\in J} q_j$ to be 
consistent with the measurement scheme $O=\prod_{i\in I} o_i$ in the sense of \cref{lem:lem1}, each operator of $q_i$ its Pauli string form with support in $I$ must be $o_i$: this introduces constraints to the set of 
consistent operators. These constraints translated into the language of $\mathbb{Z}_2^{2N}$ are linear: if $X_i$ is measured at site $i$, the additional constraint is $z_i=0$ (this encompasses the cases of $X_i$ and identity at site $i$), if $Z_i$ is measured, the constraint is $x_i=0$, and for $Y_i$ it is $x_i+z_i=0 \mod 2$. All of them are independent, and there are $\lvert I \rvert$ of them, one for each measured site. Thus, the linear subspace of $\mathbb{Z}_2^{2N}$ corresponding to measurement-consistent stabilizers has dimension $N-\lvert I\rvert$; any basis of this subspace corresponds to generators of the post-measurement state stabilizer.
\end{proof}
For example, the correlations that stem from the measurement described after \cref{lem:lem1} and in \cref{fig:fiveqbpath} correspond to the following $\vec x \oplus \vec z$ vectors in the binary representation:
\begin{equation-aligned}
    (0, 1, 0, 1)&\oplus(1, 0, 0, 0)& {\color{gray}\approx}&~~{\color{gray} Z_1 X_2 X_4,}\\
    (1, 0, 1, 0)&\oplus(0, 0, 0, 1)& 
    {\color{gray}\approx}&~~{\color{gray} X_1 X_3 Z_4,}\\ 
    (1,  1, 1, 1)&\oplus(1, 0,  0, 1)& {\color{gray}\approx}&{\color{gray} -}{\color{gray} Y_1 X_2 X_3 Y_4.}
\end{equation-aligned}
The third vector is a linear combination of the first two; thus, the space of end correlations is two-dimensional.

\subsection{Discussion of measurement outcomes}
\label{app:outcomes}

\begin{figure}
\centering
\includegraphics[width=.8\linewidth]{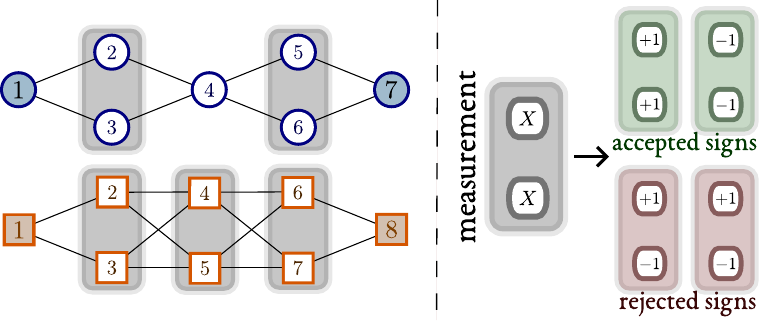}
\caption{Some measurement patterns have a restricted set of accepted outcomes. If we measure the pairs of qubits associated with the vertices in the gray areas in the Pauli $X$ basis, the products of the outcomes have to be equal to one. This indicates which combinations of signs are allowed: the accepted signs are shown in the green areas, while the rejected signs are shown in the red areas. We postselect by discarding all states whose measurement outcomes have signs that disagree with the target graph state.}
\label{fig:acc-signs}
\end{figure}

In the exemplary cases of \cref{sec:noise_extract}, we measure qubits of graph states in Pauli measurement bases which agree with the states' stabilizer operators. If the stabilizer operator $Q$ is fully embedded in the measurement structure, the product of the measurement outcome is known (see \cref{cor:obsstruct}). This reduces the allowed combinations of single measurement outcomes by one-half. For example, the twisted graph  $G$ in \cref{fig:acc-signs} is stabilized by $Q = g_2 g_3 = X_2 X_3$. We know, that $\expval{Q}{G} = s_2 s_3 = 1$, where $s_2, s_3$ are the single measurement outcomes.
It follows that the outcomes have to be equal, that is $s_2 = s_3$. If it nevertheless happens that we measure the two qubits of the state in the bases  $X_2 X_3$ and receive outcomes which are different, e.g. $(s_2, s_3)=(1, -1)$, it indicates that we measured a different state (assuming the measurement itself is perfect).  The same argument can be made for all gray regions in \cref{fig:acc-signs}.

In \cref{sec:noise_extract,sec:GHZ}, we perform measurements on noisy graph states and postselect on measurement outcomes. That is, we consider all stabilizer operators with support within the measurement pattern and check the outcome for consistency. If the outcomes do not match the expectation, we discard the resulting state.
This concept is also pictured in \cref{fig:fourqubit}.

\subsection{Patterns for specific graph families} \label{app:patterns}

Consider a path graph $G$ of $n+2$ vertices $V=\{0,1,\ldots,n,n+1\}$ with edges $E=\{\{k,k+1\} \}_{k=0}^n$. If the qubits $\{1,\ldots,n\}$ of the associated state $\ket{G}$ are measured in the basis of $X$, this induces correlations in the remaining unmeasured qubits $0$ and $n+1$. This follows from the fact that the following stabilizer operators are consistent with the measurement pattern:
\begin{equation-aligned}
    Z_0 Z_{n+1} \prod_{k\in I} X_k,& ~~~X_0 X_{n+1} \prod_{k\in I'} X_k,& (\text{odd }n)\\
    Z_0 X_{n+1} \prod_{k\in I} X_k,& ~~~X_0 Z_{n+1} \prod_{k\in I'} X_k.& (\text{even }n)\\
    \label{eqn:stabextended}
\end{equation-aligned}
where $I$ ($I'$) is the set of odd (even) numbers between $1$ and $n$:
\begin{equation-aligned}
    I &= \{ k \in \{1,2,\ldots,n\} : k=1 \mod 2\},\\
    I' &= \{ k \in \{1,2,\ldots,n\} : k=0 \mod 2\}.
\end{equation-aligned}

Thus, if the outcome of the $X_k$ measurement is denoted by $s_k$, the following operators generate the post-measurement stabilizer of the terminal qubits $\{0,n+1\}$:
\begin{equation-aligned}
      Z_0 Z_{n+1} \prod_{k\in I} s_k,& ~X_0 X_{n+1} \prod_{k\in I'} s_k,& (\text{odd }n)\\
    Z_0 X_{n+1} \prod_{k\in I} s_k,& ~X_0 Z_{n+1} \prod_{k\in I'} s_k.& (\text{even }n)\\
    \label{eqn:termoddpathcorr}
\end{equation-aligned}
Both cases are consistent with a maximally entangled two-qubit state being produced in the terminal qubits: for the odd-$n$ case, the existence of stabilizers described by  \cref{eqn:termoddpathcorr} ensures that further measurement of $X$, $Y$ or $Z$ at qubit $0$ is perfectly (anti)correlated with the outcome of measurement of the same operator at qubit $n+1$.

This result generalizes to the case where each qubit $k$ is replaced by a set of $m_k$ qubits in such a way that the adjacent layers are fully connected. In this case, the qubits are denoted by pairs of numbers $(k,i)$, where $k=\{0,1,\ldots,n,n+1\}$ is the layer index and $i\in\{1,\ldots,m_k\}$ is the qubit index within the layer. The edges exist between any layer-adjacent qubit: 
\begin{equation}
E=\{\{(k,i),(k+1,j)\}:k\in\{0,\ldots,n\},i\in\{1,\ldots,m_k\},j\in\{1,\ldots,m_{k+1}\}\}.
\end{equation}
This class of graphs includes the path graph (each layer has only a single qubit), the crazy graph (the layers $0$ and $n+1$ contain one qubit each, every other has two), and twisted graphs (the qubit count alternates between 1 and 2) as well. Suppose that all the qubits in every internal ($i=1, \ldots, n$) layer are measured in the $X$ basis. The stabilizer operators that determine the remaining correlations must be consistent with this structure. A particularly simple set of such operators exists: first, observe that if a layer $k$ contains at least two qubits $i$ and $i'$, then $X_{k,i}X_{k,i'}$ is a stabilizer operator. Thus, the signs of outcomes have to be equal within a layer. Therefore, for any choice of inner-layer qubit index $(i_k)$ the direct analogue of \cref{eqn:stabextended} holds:
\begin{equation-aligned}
    Z_0 Z_{n+1} \prod_{k\in I} X_{k,i_k},& ~X_0 X_{n+1} \prod_{k\in I'} X_{k,i_k},& (\text{odd }n),\\
    Z_0 X_{n+1} \prod_{k\in I} X_{k,i_k},& ~X_0 Z_{n+1} \prod_{k\in I'} X_{k,i_k},& (\text{even }n),\\
    \label{eqn:stabextendedgeneral}
\end{equation-aligned}
are stabilizer operators.
For symmetry, we assume the twisted graph structure exists only for an odd number of layers, so that the terminal ones consist only of one qubit each. Thus, the terminal layers $0$ and $n+1$ always consist only of one qubit each, and the inner indices are omitted.
Consequently, if $s_k$ in \cref{eqn:termoddpathcorr} are interpreted as signs of outcomes at layer $k$ (which are all equal within a layer due to the structure of the stabilizer operators mentioned previously), the same post-measurement stabilizer structure appears. 

In the case of the crazy graph, an additional observation can be made: loss of any single internal (not adjacent to the terminal qubits) edge yields the same post-measurement stabilizer structure, since the indices $(i_k)$ in \cref{eqn:stabextendedgeneral} can be chosen arbitrarily. Thus, the crazy graph family is especially tolerant to single-edge loss: it does not affect the post-measurement state, if the outcomes are postselected on the ideal sign structure. This can be generalized to the loss of any subset of edges not adjacent to each other and the terminal qubits.

\section{Derivation for the noise models} \label{app:noise}

\subsection{Uncorrelated phase noise} \label{app:uncorr_phase_noise}

\begin{figure}\centering \includegraphics[width=.4\linewidth]{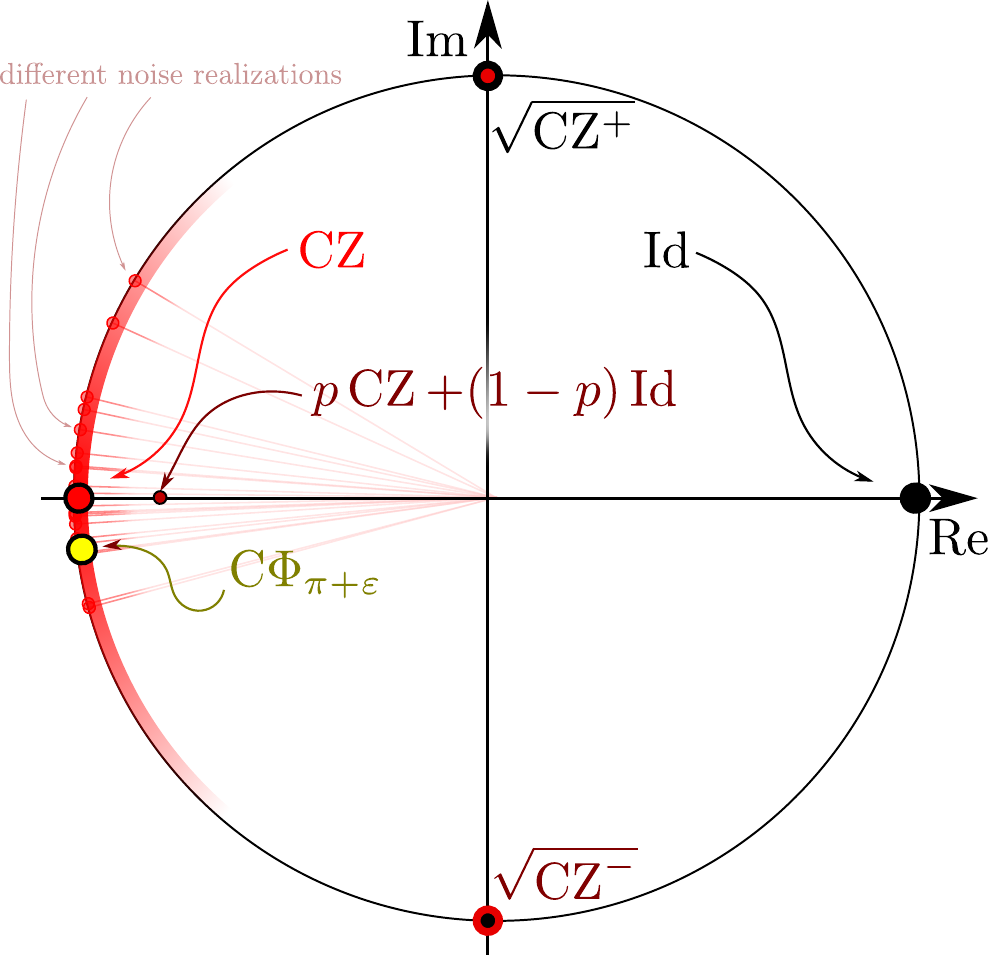}
\caption{The controlled-$Z$ gate, if implemented as a controlled phase operation, is susceptible to phase variations. In each experimental run the phase is constant, and the resulting state is a weighted graph state. Since this phase is unknown, the effective quantum state is a weighted graph state, averaged over different noise realizations. For just a single edge, if the phase distribution is centered around $\pi$ (corresponding to the ideal graph state), the resulting quantum channel is a mixture of controlled-$Z$ and identity: the edge is probabilistically generated. }
\label{fig:cz}
\end{figure}

In certain ion traps and superconducting quantum devices, the sequence of $\CZ$ gates is realized with evolution with an Ising-type Hamiltonian \cite{Stehlik2021,long2021universal,wunderlich2009two,piltz2014trapped} :
\begin{equation}
    H=\sum_{(i,j)\in G} \gamma_{i,j} \frac{Z_i+1}{2}\frac{Z_j+1}{2}.
    \label{eq:isingham_app}
\end{equation}
For constant $\gamma=\gamma_{i,j}$, the unitary operation $\exp(-i H t)$ exactly corresponds to the product of $\CZ$ gates appearing in  \cref{eq:grstate} for $t=\frac{\pi}{\gamma}$. However, incomplete control over the system leads to a fluctuation of the interaction strengths $\gamma_{i,j}$. In such a case, each of the two-body terms appearing in \cref{eq:isingham_app} gives rise to a controlled phase operation $\CP_{i,j} = \diag (1,1,1,\exp(-it\gamma_{i,j}))$ in the appropriate two-body subsystem. The quantum channel corresponding to this gate decomposes:
\begin{align}
    \CP \rho \CP^\dagger= \frac{1-\cos\phi}2 \Cz[\rho] + \frac{1+\cos\phi}{2}\rho+\frac{\sin\phi}2\CS^+[\rho] -\frac{\sin\phi}2 \CS^-[\rho],
\label{eq:cpdecomp}
\end{align}
where $\phi=t\gamma_e $ and $\CS^\pm$ is defined in  \cref{eq:cschan}. If the actually realized phase cannot be measured classically in each experimental run, the resulting state is a mixture of states:
\begin{equation}
    \rho = \int P(\{\gamma_{i,j}\}) \ket{G_\gamma} \bra{G_\gamma} \dd \gamma,
    \label{eq:defmix_app}
    \end{equation}
where $P(\{\gamma_{e}\})$ denotes the probability distribution of phases and $\ket{G_\gamma}=\left(\prod_{i,j} \CP_{i,j}\right) \ket{+}^{\otimes N}$. If the strength noise is independent for each edge (i.e. $P=\prod P_{e}$) and symmetric around $\pi$, a significantly easier description of the state $\rho$ can be found: each controlled-phase averages to a mixture of controlled-$Z$ and identity (see  \cref{fig:cz}), since these are the only terms appearing in  \cref{eq:cpdecomp} with prefactors symmetric around $\pi$. Thus, noise effectively leads to the generation of \emph{randomized graph states}: this is an ensemble of graph states, where the edge $e\in E$ is missing with probability $p_{e}$:
\begin{equation}
    \rho = \sum_{E' \subset E} \left(\prod_{{\{i',j'\} \in E' }}\!\!\! (1-p_{i',j'})\!\!\! \prod_{\{i,j\}\notin E'}\!\!\! p_{i,j}\right)\ket{G'}\bra{G'},
    \label{eq:randomgraph_app}
\end{equation}
where the graph $G'$ has the edges defined by the set $E'$. The probability $p_{i,j}$ can be determined by the statistical properties of the distribution of the phase $\phi$ associated with the edge $\{i,j\}$. In particular, for the normal distribution of $\varepsilon$ centered around $\pi$  with the standard deviation of $\sigma$ the probability $p$ takes the closed form of 
\begin{equation-aligned}
    p = \frac{1-\exp\left(-\frac{\sigma^2}2\right)}2=\frac{\sigma^2}4-\frac{\sigma^4}{16}+O(\sigma^6).
    \label{eq:pasphase}
\end{equation-aligned}

\subsection{Correlated phase noise} \label{app:corr_phase_noise}
Correlated phase noise has a similar, but more involved effect. Suppose that the weights in Hamiltonian \cref{eq:isingham_app} are all the same, but the interaction time varies. In such a case, the phase is the same for each controlled-phase operator, and the decomposition used to derive \cref{eq:randomgraph} breaks. This is due to the fact that in the integrand of \cref{eq:defmix_app} there exist additional terms symmetric around the ideal phase of $\phi=\pi$. There are the terms proportional to $\sin\phi$ in  \cref{eq:cpdecomp}: even powers of them also contribute to the final state. Thus, in the lowest order or approximation, the output state is described by:
\begin{itemize}
    \item The unmodified state graph state $\ket{G}$, with reduced probability coming from $\Cz$ terms appearing in product of channels defined by  \cref{eq:cpdecomp},
    \item graph states with one edge missing, as in the case of uncorrelated noise,
    \item graph states modified by a product of two $\CS^\pm$ channels corresponding to different edges, with possible negative weights coming from negative signs in  \cref{eq:cpdecomp}.
\end{itemize}
In the lowest order of approximation, corresponding to the $\sigma^2$ terms in \cref{eq:pasphase}, the calculation yields the final state (after averaging) of the form
\begin{equation-aligned}
\rho =& (1-p\lvert E\rvert)\ket{G}\bra{G} + p\sum_{e\in E}\lvert{\!\!\overbrace{\!G'\!}^{G\setminus\{e\}}\!\!}\rangle \bra{G'}\\
&+\frac{p}{2}\sum_{(e,e'\in E}^{e\neq e'} \sum_{s_1, s_2 = \pm1} s_1 s_2 \CS_e^{s_1}  \CS_{e'}^{s_2} [\ket{G}\bra{G}] .
\end{equation-aligned}
Here, $\lvert E\rvert$ denotes the number of edges in the graph $G$. Note that this decomposition is not convex and cannot be interpreted as an ensemble: the minus signs in front of the decomposition terms prevent this. Signs appear as a result of the discrete decomposition of the controlled-phase and subsequent averaging; see \cref{app:noise} for details. In this decomposition, $G-e \coloneqq (V,E\setminus \{e\})$ is the graph with edge $e$ removed, and ${\CS^\pm}$ is the application of one of the two unitary square roots of $\CZ$:
\begin{equation}
    {\CS^\pm} [\rho] = \sqrt{\CZ^\pm} \rho \sqrt{\CZ^\pm}^\dagger,
    \label{eq:cschan}
\end{equation}
where
\begin{align}
    \sqrt{\CZ^\pm} \coloneqq \frac{1}{2} \left[ \left( 1 + e^{\pm i \frac{\pi}{2}} \right) \1 + \left( 1 - e^{\pm i \frac{\pi}{2}} \right) \CZ \right] = \begin{pmatrix} 1 &&&\\&1&&\\&&1&\\&&&\pm i\end{pmatrix}.
    \label{eq:sqrtczop}
\end{align}

This decomposition is not unique, and in principle, different weights for the full, edge-missing, and $\CS$-modified graph states are possible; however, some of the signs will always be negative, regardless of the choice. All such decompositions are only low noise approximations for $p\ll \lvert E\rvert^{-1}$.
\subsection{Fusion gate noise}
\begin{figure}
\centering\includegraphics[width=.6\linewidth]{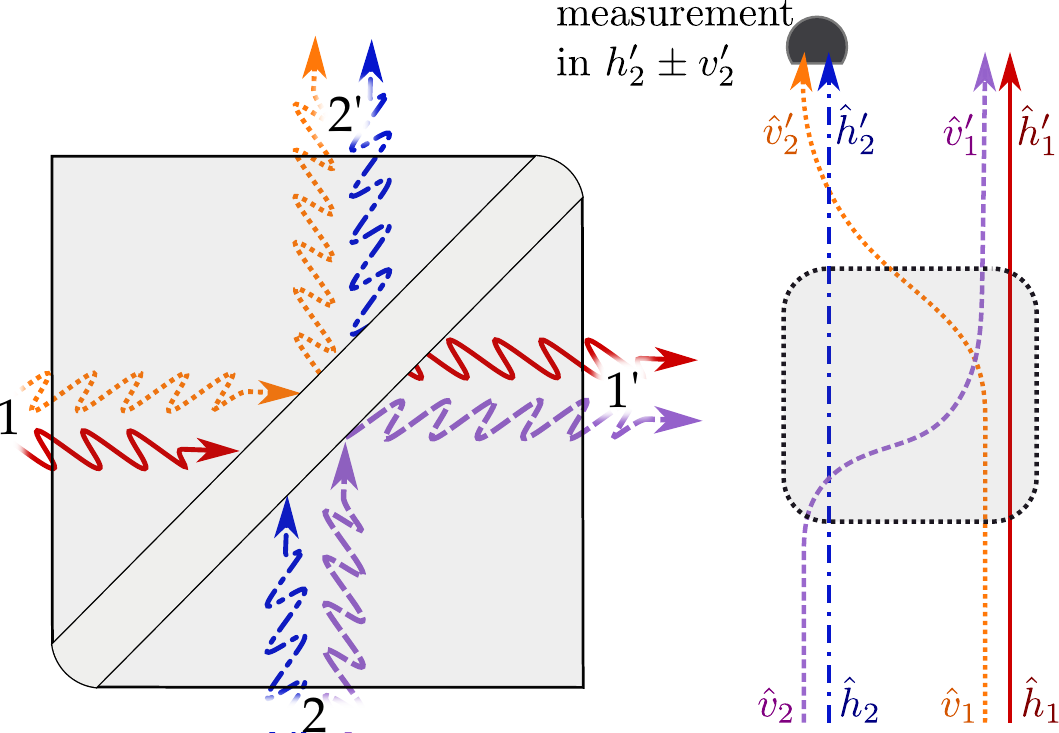}
\caption{Polarizing beamsplitter transmits horizontally polarized photons and reflects vertically polarized ones, entangling the input state in the output modes. Subsequent polarization measurement in the diagonal basis of one of the output modes is a basis for a type 1 fusion gate, which is used in the generation of larger graph states from Bell pairs. }
\label{fig:polbeamsplit}
\end{figure}
\newcommand{\bellp}{{\ensuremath\bullet\!\!-\!\!\bullet}}

A different scheme for generating large cluster states is employed in linear optics experiments \cite{Lu2007}. In such systems, the qubits are encoded as photons in different light modes. Two modes are associated with each qubit, and a single photon being in one of them corresponds to the standard computational basis states. Then, the graph states are built by qubit fusion of initially generated Bell pairs. If a type 1 entangling gate is used for this, imperfect mode matching leads to the affected qubits be effectively depolarized in the $Z$ basis. This has the exact same result on the output state as local (single-qubit) phase noise, despite different physical origins.

Usually, the two modes defining qubit degrees of freedom are polarization (horizontal/vertical) modes with overlapping spatial profiles. Thus, let us denote the creation operators of those modes by $\hhd_i$ and $\hvd_i$.
In this context, a Bell pair encoded as the nontrivial two-vertex graph state can be written as
\begin{equation}
    \ket\bellp=\frac12(1,1,1,-1)^T = \frac{\hhd_{1} \hhd_{2}+\hvd_{1} \hhd_{2}+\hhd_{1} \hvd_{2}-\hvd_{ 1} \hvd_{2}}{2} \ket\emptyset,
\end{equation}
where $\ket\emptyset$ denotes the vacuum state of the optical field. 
For the creation of a larger graph states, $N$ such  Bell pairs are generated through spontaneous parametric downconversion, with the initial state of
\begin{equation}
     \ket{\bellp}^{\otimes N}=\prod_{i=1}^{N}\frac{\hhd_{2i-1} \hhd_{2i}+\hvd_{2i-1} \hhd_{2i}+\hhd_{2i-1} \hvd_{2i}-\hvd_{2i-1} \hvd_{2i}}{2}\ket\emptyset.
\end{equation}

The Bell pairs are subsequently fused to create a larger graph state. Every fusion takes two spatial modes (indexed by numbers in the subscription) and mixes them through a polarizing beamsplitter -- see \cref{fig:polbeamsplit} for the details. Then, one of the output modes is measured in the diagonal $h\pm v$ polarization basis. To see the effect of such an operation, consider an arbitrary input graph state (not necessarily a tensor product of Bell pairs) for graph $G=(V,E)$ such that the vertices 1 and 2 are not connected:
\begin{equation-aligned}
    \ket G&=\prod_{\{1,j\}\in E}\CZ_{1,i} \prod_{\{2,j\}\in E} \CZ_{2,j} \ket{+}\otimes\ket+\otimes\ket{G'}\\
    &=\frac{\hhd_1 \hhd_2 \hat f_{++} +\hvd_1 \hhd_2 \hat f_{-+} +\hhd_1 \hvd_2 \hat f_{+-} +\hvd_1 \hvd_2 \hat f_{--} }{2}\ket\emptyset,
\end{equation-aligned}
where the symbols $\hat f$ denote expressions involving creation operators to arrive at the state $\ket\psi$. The $\hat f_{++}$ prepares the graph state $\ket{G'}$, while $\hat f_{-+}$ creates the state $\prod_{(1,i)\in G} z_i\ket{G'}$: it is the postselected graph state upon measuring negative sign of the first qubit in the $Z$ basis. The other two operators work similarly; however, $\hat f_{--}$ flips the phase of $\ket{G'}$ only in the qubits that are connected to only $1$ or $2$.

Polarizing beamsplitter, described mathematically by a unitary operator $U$, exchanges the vertical polarization, keeping the horizontal ones intact ($U \hhd_{i} = \hhd_{i} U$, $U \hvd_{1} = \hvd_{2} U$, and $U \hvd_{2} = \hvd_{1} U$). First part of the fusion process leads to:
\begin{align}
    U\ket G&=\frac{\hhd_1 \hhd_2 \hat f_{++} +\hvd_2 \hhd_2 \hat f_{-+} +\hhd_1 \hvd_1 \hat f_{+-} +\hvd_1 \hvd_2 \hat f_{--} }{2}\ket\emptyset.
\end{align}
The spatial mode 2 is subsequently measured in the basis of $(\hhd_2\pm\hvd_2)\ket\emptyset$, corresponding to the $\ket\pm$ of the logical qubit. If only a single $h+v$ photon is detected, the postselected output state reads
\begin{align}
    \frac{\hhd_1 \hat f_{++} +\hvd_1 \hat f_{--} }{\sqrt2}\ket\emptyset.
\end{align}
This is exactly the state corresponding to a graph $G''=(V'',E'')$ for which vertices 1 and 2 were merged with duplicate edges removed: $E''=E'\cup \{\{1,i\} : \{1,i\} \in E \veebar \{2,i\} \in E\}$.

This is the case only if the photons in modes 1 and 2 are indistinguishable after mixing through the polarizing beamsplitter -- the procedure relies on the Hong-Ou-Mandel effect. Should this assumption not be met, the interference required for the fusion to work does not take place, and effectively the two photons originating in modes 1 and 2 are measured independently. If only one photon was observed in the new mode 2, half of the time it arrived from the mode (1,v), while the photon from mode 2 is now localized in the new mode (1,v) and the final state is $\hvd_1 \hat f_{--} \ket\emptyset$. The opposite process means that the resulting state $\hhd_1 \hat f_{++} \ket\emptyset$. 

This effectively creates an ensemble which allows for the following interpretation: the fusion took place as if the photons were indistinguishable, but it was immediately measured in the $Z_1$ basis and the result of the measurement is not known. This is equivalent to a simple probabilistic $Z_1$ flip, described by the following channel:
\begin{align}
\rho \mapsto \overbrace{\Pi_{Z_1+}}^{\frac12(1+Z_1)} \rho \Pi_{Z+}  + \overbrace{\Pi_{Z_1-}}^{\frac12({1-Z_1})} \rho \Pi_{Z_1-} = \frac{\rho + Z_1 \rho Z_1}{2}.
\end{align}
The $Z$-flip ensemble is equivalent to the $Z$-measurement one and is easier to work with. It is easier to observe that the noise does not propagate if the affected qubit is then fused with others, since there is always just a (potentially $Z$-flipped) graph state at the input -- unless it also fails because the noise is correlated.

\section{Cancellation of correlated edge noise effects in twist and crazy graph families}
\label{app:corredgecancel}

If for the square graph presented in  \cref{fig:fourqubit}a,  instead of the perfect controlled-$Z$ unitary, the controlled phase \begin{equation-aligned}
    \CP(\pi+\varepsilon)=\diag(1,1,1,-1\exp(i\varepsilon))
\end{equation-aligned} is used {with the same $\varepsilon$ for all edges, the result is a four-qubit \emph{weighted graph state} which we denote by} $\ket{G}_\varepsilon$. If qubits 1 and 3 are measured in the $X$ basis, {we postselect on the equal measurement outcomes.}  With the weighted graph state $\ket{G}_{\varepsilon}$ as the initial state, let us denote the unnormalized post-measurement state by
\begin{align}
    \ket{\xi}_\pm = \frac{\mathbbm{1}\pm X_1}2\frac{\mathbbm{1}\pm X_3}2\ket{G}_\varepsilon.
\end{align}

Note the equality of signs $\pm$ appearing in two projection operators: we postselect on observing this type of outcome, since unequal signs can not appear in the ideal ($\varepsilon=0$) case. The relevant expectation values defining the end correlations (with $\langle O \rangle_\pm$ denoting the expectation value of $O$ for the unnormalized state $\ket{\xi}_\pm$) across qubits $2$ and $4$ are the following, in the lowest nontrivial noise order:
\begin{align}
    \langle Z_2 Z_4\rangle_+ &= 1-\varepsilon^2 +O(\varepsilon^4), & \langle Z_2 Z_4\rangle_- &= -1 +\frac{\varepsilon^2}{2} +O(\varepsilon^4), \notag \\
    \langle X_2 X_4 \rangle_+ &= 1-3\varepsilon^2+O(\varepsilon^4), & \langle X_2 X_4 \rangle_- &= -1 +\frac{\varepsilon^2}{2}+O(\varepsilon^4), \notag \\
    \langle Y_2 Y_4 \rangle_+ &= -1+3\varepsilon^2+O(\varepsilon^4), & \langle Y_2 Y_4 \rangle_- &= 1 -\frac{\varepsilon^2}{2}+O(\varepsilon^4), \notag \\
    \| \vert \xi\rangle_+ \|^2 &= 1-\varepsilon^2+O(\varepsilon^4), & \| \vert \xi\rangle_- \|^2 &= 1 -\frac{\varepsilon^2}{2}+O(\varepsilon^4). 
\end{align}
Note that for the $\ket{\xi}_-$ outcome, the squared length $\| \vert \xi\rangle_- \|^2$ is proportional to the unnormalized expectation values: thus, for the physical correlations the dominant noise effects are canceled and only higher order terms ($O(\varepsilon^4)$) remain. On the other hand, if the outcome is $\ket{\xi}_+$, the noise is amplified compared to postselection on only stabilizer-consistent outcomes.  Thus, further postselection can amplify the entanglement quality of the remaining qubits at the cost of discarding certain outcomes and reducing production rates.

The double twist graph can be viewed as multiple such squares stacked together by the terminal vertices: thus, a proper postselection can mitigate the first nontrivial noise effects completely (see \cref{fig:corredgesus}). A similar structure appears for the crazy graph: it is much more complex to analyze, but computer algebra systems provide evidence of a periodic structure (also visible in \cref{fig:corredgesus})   of susceptibilities as the length of the graph increases.

\bibliographystyle{quantum}
\bibliography{references}

\end{document}